\pdfoutput=1
\documentclass[twoside]{article}

\usepackage[accepted]{aistats2025}

\usepackage{url}            
\usepackage{booktabs}       
\usepackage{amsfonts}       
\usepackage{nicefrac}       

\usepackage[numbers]{natbib} 

\usepackage{hyperref}
\usepackage{xcolor}

\usepackage{comment}
\usepackage{subfigure}
\usepackage{amsmath}
\usepackage{amssymb}
\usepackage{mathtools}
\usepackage{amsthm}
\usepackage{tablefootnote}
\usepackage{mathrsfs}
\usepackage[capitalize,noabbrev]{cleveref}
\usepackage{wrapfig}
\usepackage{algorithm}
\usepackage[algo2e, ruled]{algorithm2e}
\usepackage{multicol}
\SetKwInput{KwData}{Input}
\SetKwInput{KwResult}{Output}

\def \bP {\mathbb{P}}
\def \bE {\mathbb{E}}

\def \bN {\mathbb{N}}

\def \cM {\mathcal{M}}

\def \var {\mathsf{Var}}
\def \cov {\mathsf{Cov}}

\usepackage{xspace}

\newcommand{\prob}[1]{\mathbb{P}\left[#1\right]}

\newcommand{\diff}{\mathrm{d}}

\newcommand{\calW}{{\mathcal{W}}}
\newcommand{\calX}{{\mathcal{X}}}

\newcommand{\bse}{\boldsymbol{e}}

\newcommand{\bsx}{\boldsymbol{x}}
\newcommand{\bsC}{\boldsymbol{C}}
\newcommand{\bsW}{\boldsymbol{W}}
\newcommand{\bss}{\boldsymbol{s}}
\newcommand{\bsv}{\boldsymbol{v}}

\newcommand{\eps}{\varepsilon}
\newcommand{\Wdp}{\calW_{\varepsilon,\delta}}
\newcommand{\err}{\mathsf{err}}

\theoremstyle{plain}
\newtheorem{theorem}{Theorem}[section]
\newtheorem{prop}[theorem]{Proposition}
\newtheorem{lemma}[theorem]{Lemma}
\newtheorem{corollary}[theorem]{Corollary}
\theoremstyle{definition}
\newtheorem{definition}[theorem]{Definition}

\theoremstyle{remark}
\newtheorem{remark}[theorem]{Remark}

\makeatletter
\renewcommand*\env@matrix[1][\arraystretch]{%
  \edef\arraystretch{#1}%
  \hskip -\arraycolsep
  \let\@ifnextchar\new@ifnextchar
  \array{*\c@MaxMatrixCols c}}
\makeatother

\begin{document}

\twocolumn[

\aistatstitle{Differentially Private Range Queries with Correlated Input Perturbation}

\aistatsauthor{ Prathamesh Dharangutte \And Jie Gao \And  Ruobin Gong \And Guanyang Wang }

\aistatsaddress{ Rutgers University \\ \texttt{\{prathamesh.d, jg1555, ruobin.gong, guanyang.wang\}@rutgers.edu} } ]




\begin{abstract}
  This work proposes a class of differentially private mechanisms for linear queries, in particular range queries, that leverages correlated input perturbation to simultaneously achieve unbiasedness, consistency,  statistical transparency, and control over utility requirements in terms of accuracy targets expressed either in certain query margins or as implied by the hierarchical database structure. The proposed Cascade Sampling algorithm instantiates the mechanism exactly and efficiently. Our theoretical and empirical analysis demonstrates that we achieve near-optimal utility, effectively compete with other methods, and retain all the favorable statistical properties discussed earlier. 
\end{abstract}

\section{Introduction}\label{sec:intro}


In this paper we construct a class of differentially private mechanisms for linear queries, including range queries, representable as a multiplicative operation of a pre-specified workload matrix and a confidential database.
Our work is strongly motivated by the application of differential privacy to the 2020 U.S. Decennial Census, providing redistricting (P.L. 94-171) as well as Demographic and Housing Characteristic (DHC) files in the forms of multi-resolutional tabular data \citep{abowd2022topdown}. Population tabulations across  geographic resolutions follow a hierarchical system termed the ``spine'' \citep{cumings2024geographic}, which orders from top to bottom geographic entities (states, counties, tracts, block groups, and blocks), with higher-level geographies partitioned by the lower-level ones. As the only high-profile deployment of differential  privacy in the public domain, the particular demands from the Census and similar official data products reveal a number of crucial considerations on privacy mechanisms that are possibly shared in a broader set of practical application scenarios.
We discuss these considerations, highlighting consistency and utility control as the most prominent. 

{\bf Consistency (internal)}: The sanitized output may plausibly be viewed as having been queried directly from an input database without modification. In particular, the data output for a geographical range (e.g. a state) should be precisely the sum of data values from its constituent sub-ranges (e.g. all the counties in a state). This is an example of a broader family of logical consistency that ensures stability and absence of conflicts in the data output. 

\vspace*{-1.5mm}
{\bf Fine-grained utility control}: The mechanism accommodates custom, externally specified utility requirements expressed as accuracy targets in query margins or implied by the hierarchical database structure. For example, Census tabulations at lower and intermediate geographies, as does certain ``off-spine'' geographies (e.g. voting districts), must meet accuracy targets according to the relevant operational standards \citep{census2022qualitymetrics}. Moreover, population counts across larger geographical regions at a lower resolution may not be permitted to have a greater error margin compared to smaller geographical regions at a higher resolution. For example, the mean error and mean absolute errors of total population counts in the Census DHC files remain consistent at the state, county, tract, and block group levels \cite{census2023detailedsummarymetrics}. With the exception of a number of very recent work~\cite{Xiao2021-vi,xiao23advances}, fine-grained utility control has been scarcely discussed in the DP literature, as the focus has been predominantly placed on the  assessment of overall utility (such as average or worst-case).
The Census application raised this issue in an interesting angle -- that the algorithm designer is given both custom specified, fine-grained utility targets as well as privacy budget target, and must work backward to meet both objectives. 



Furthermore, unbiasedness and statistical transparency influence both the quality and usability of the data product \citep{fioretto2019differential,gong2022transparent}. With {\bf unbiasedness}, the sanitized output exhibits no bias with respect to the ground truth;
With {\bf statistical transparency}, the probabilistic description of the sanitized output is analytically tractable (ideally in a closed-form) to enable reliable downstream statistical inferences. Last, it is always desirable to have 
{\bf efficient implementation} -- the algorithm is exact and simple to implement, with no need for approximate simulation (e.g. Markov chain Monte Carlo) nor optimization-based post-processing. 
We will review the Census Bureau's \emph{TopDown algorithm}~\citep{abowd2022topdown} as well as other DP mechanisms with respect to these considerations later.  

This paper considers \emph{input perturbation} -- adding Gaussian noises at each input data items and answering queries with the perturbed data. Classic input perturbation mechanisms naturally support unbiasedness, logical consistency, and statistical transparency, and are practical to implement. However, they do not support fine utility control and typically result in poor data utility, an issue that worsens when the query range contains a large number of data elements~\citep{chan2011private}. 
For this reason, {input perturbation} methods have been largely under-utilized in practice.

In this work we consider input perturbation with \emph{correlated Gaussian noise}, which reduces error magnitude and offers fine control over utility
while harnessing its many strengths. The proposed mechanism carefully couples the item-wise noises 
to allow queries at different hierarchical levels to conform to a \emph{uniform} accuracy standard while achieving \emph{near-optimal} overall utility objectives, both theoretically and empirically. 
We show for 1D range queries, the proposed mechanism achieves optimal mean square error and near-optimal worst-case and expected worst-case errors when compared to prevailing $(\epsilon, \delta)$-DP mechanisms.
The special error correlation structure supports a linear time efficient implementation called the Cascade Sampling algorithm. Importantly, the fine control over data utility at different levels of geography is inherent to the design of the proposed mechanism, rather than reliant on optimization-based post-processing which may destroy transparency and render unpredictable accuracy. Our proposal generalizes to other hierarchical and multidimensional linear query settings. 



\section{Problem Formulation}\label{sec:definition}

Given a confidential data vector $\bsx$ of dimension $n$, and a workload matrix $\bsW$ of dimension $p\times n$, we would like to report a (possibly) noisy version of the query answer $\bsW \bsx$ while preserving the privacy of individual data elements in $\bsx$. 
$\bsW$ is an incidence matrix with rows corresponding to queries and columns corresponding to data elements. Specifically, we consider an $(\eps, \delta)$-differentially private mechanism $\Wdp$ which satisfies for any two neighboring databases $\bsx, \bsx'$, $||\bsx-\bsx'||_1\leq 1$, and any set $D$ of output values $$\prob{\Wdp(\bsx)\in D}\leq e^{\eps}\cdot \prob{\Wdp(\bsx')\in D}+\delta.$$
The linear query framework models many scenarios in practice. Three are particularly relevant to this work. 
{\it Predicate counting queries} report the number of database rows that satisfy the given predicate $q$, which are encoded into the rows of the workload matrix $\bsW$.  
{\it Range queries} report the sum of elements (or coordinates; $x_i$) that fall inside a given range, such as a time interval $[\ell, r]$ (e.g. streaming data) or a two-dimensional geographic area, with the structure of the range reflected in $\bsW$.
 {\it Contingency tables} are multidimensional histograms of entities satisfying certain composite attributes in a database. They can be regarded as a special type of high-dimensional range query, and are used extensively by statistical agencies for data processing and dissemination.

\subsection{Definitions}\label{sec:properties}





\begin{definition}
A mechanism $\Wdp$ is \emph{unbiased} if
 $\bE(\Wdp(\bsx)) = \bsW\bsx,$
 where the expectation is taken over the randomness of $\Wdp$.
\end{definition}
That is, unbiasedness forbids a privacy mechanism from injecting  systematic drift into the data output. 

\begin{definition}[\cite{gong2022transparent}, Def. 3] A privacy mechanism $\mathcal{W}$ is \emph{statistically transparent} if the conditional distribution of its output given the input,
$p_{\xi} ( \mathcal{W} = w \mid \bsx = x ),$
is analytically available up to $p$ and $\xi$, where $\xi$ is the parameter for $p$ (both tuning and auxiliary).
\end{definition}

Statistical transparency is not frequently discussed in the literature of private mechanism design, but it is crucial if the sanitized output is subject to further data analysis as
it provides the basis for valid statistical  uncertainty quantification \citep{gong2022transparent}.

\begin{definition}
A mechanism $\Wdp$ operating on the data vector $\bsx$ is internally \emph{consistent} if with probability one (over the randomness of $\Wdp$) there exists a vector $\bsx'$ such that $\Wdp(\bsx)=\bsW\bsx'$.
\end{definition}

First defined for contingency tables and generalizable to any query, consistency requires the sanitized query to be a legitimate output of the intended query applied to a potential input database~\citep{barak2007privacy,hay2009boosting,chan2011private}. 
It is particularly important if the sanitized output enters directly into downstream decisions and is expected to  
be free of internal logical conflicts.
In the literature, external consistency has also been discussed. For example, state-level populations must be exactly reported per their constitutional purpose for reapportionment -- an ``invariant'' requirement enforced externally; see~\cite{gao2022subspace,dharangutte2023integer}. This is not discussed in this paper. 

\begin{remark}\label{remark:consistency}
For linear queries, consistency requires the sanitized output to be in the column space of $\bsW$. Any logical relationship embodied in $\bsW$ (e.g. one range being the union of two  disjoint ranges) is mirrored in $\Wdp$. Additive mechanisms of the form  $\bsW \bsx + \bse$ may not automatically obey consistency, unless $\bse$ is guaranteed to be in the column space of $\bsW$. 
The same is true for exponential mechanisms unless the range is intentionally restricted \citep{seeman2022partially}.
\end{remark}

\section{Correlated input perturbation mechanism}\label{sec:corr-noise}


This section presents the design of the correlation matrix, an efficient algorithm to sample from this distribution, and the resulting privacy and utility guarantees. Proof of our results are postponed to~\Cref{appendix:proofs} and extensions to general binary tree and 2-D data are discussed in \cref{subsec: generalization}.

\subsection{Correlation matrix}\label{subsec:correlation-matrix}
To simplify matters, we consider $n$ data points on a one-dimensional line and assume that $n$ is a power of 2, denoted as $n = 2^k$. Extensions  will be discussed in Section \ref{subsec: generalization}. Each data point will be represented by its binary form, utilizing $k$ bits, as leaves of a perfect binary tree with height $k$. Nodes in the tree receive labels based on their positions in a level-order traversal. 
For example, the root node is labelled as $\varnothing$, the  nodes at depth 1 will have labels ``0'' and ``1'', and so on. 

Our objective is to allocate Gaussian random noises, denoted as $\{X_I\}_{I\in\{0,1\}^k}$, to every data point (i.e. leaf node). The noise imposed on each internal node is the sum of the noises of its two children. With our labeling convention, this relationship can be succinctly represented as $X_{\star} = X_{\star 0} + X_{\star 1}$, where $\star$ stands for any binary sequence (including the empty one) with a length less than $k$. 
If all the noises on the leaf nodes are independently and identically distributed (i.i.d.), we anticipate that the noise introduced at the root  will have a Gaussian variance of $\Theta(n)$.
However, by carefully coupling the Gaussian variables on the leaf nodes, we can establish \emph{uniform} variance across \emph{all} nodes in the binary tree. The structure of our correlation matrix is recursively defined below: 
\begin{definition}\label{def:correlated-noise-correlation}
$J_i$ is the all-one matrix of size $2^i \times 2^i$. $\bsC_i$ is a square matrix of size $2^i\times 2^i$:
\begin{align}\label{eqn: corr_recursive}
\bsC_1 :=    \begin{pmatrix}1 & -\frac{1}{2} \\ -\frac{1}{2} & 1\end{pmatrix}, \;
    \bsC_{i+1} := \begin{pmatrix}
\bsC_{i} & -\frac{J_i}{2^{2i+1}} \\
-\frac{J_i}{2^{2i+1}} & \bsC_{i}
\end{pmatrix}.
\end{align}
\end{definition}
\begin{definition}\label{def:correlated-noise-mechanism}
For $n = 2^k$ data points identified by binary representation, our correlated noise mechanism  is defined as $\mathsf{Noise} = \sigma Z$. Here, $Z\sim\bN(\mathbf{0}, \bsC_k)$, and $\sigma$ depends on a later-specified privacy budget. This mechanism applies to any private data vector $\bsx$, resulting in the output $\bsx + \mathsf{Noise}$.
\end{definition}

The structure outlined in Definition \ref{def:correlated-noise-mechanism} has a clear recursive pattern. One can split the \(2^k\) data points into a left subtree, where labels begin with \(0\), and a right subtree, where labels begin with \(1\). The collection of points within each group mirrors a perfect binary tree of depth \(k-1\) with a covariance matrix of \(\sigma^2 \bsC_{k-1}\). Points belonging to different groupings have a slightly negative correlation of \(-2^{-2k+1}\). Similarly, points in the left subtree can be further divided based on those starting with \(00\) and those beginning with \(01\). Each of these smaller sub-groupings exhibits a covariance of \(\bsC_{k-2}\), and any pair of points from these groups share a correlation of \(-2^{-2k+3}\). This process can be recursively applied until each group is reduced to a single data point.

The next result shows that the variance of each internal node is the same as that of every leaf node. 

\begin{theorem}\label{thm:noise-equal-variance}
    Consider $n = 2^k$ data points identified by their binary representation as described earlier. Assuming the noise mechanism is defined as per Definition \ref{def:correlated-noise-mechanism}, every node in the binary tree, including both leaf and internal nodes, has a marginal distribution of $\bN(0, \sigma^2)$. This implies that each node shares an identical variance.
\end{theorem}

\subsection{Cascade Sampling algorithm}\label{subsec:algorithm-1d}
This section introduces an efficient algorithm of running time $O(n)$, where $n = 2^k$ represents the total number of data points, for generating samples from our defined noise mechanism, specifically $\bN(0, \sigma^2 C_k)$ as defined in Formula \eqref{eqn: corr_recursive}. This method significantly improves standard Gaussian generation methods in scalability, as supported by theoretical and numerical evidence.

Assuming $\sigma = 1$ without loss of generality, the standard method for sampling an $n$-dimensional Gaussian $\bN(\mathbf{\mu}, \Sigma)$ involves Cholesky decomposition of the covariance matrix $\Sigma = L L^\top$, followed by transforming a standard Gaussian vector ($\mathbf{x}\sim \bN(0, \mathbb I_n)$) using $L\mathbf{x} + \mathbf{\mu}$. This process has a high computational cost, primarily due to the $O(n^3)$ expense of the Cholesky decomposition, with additional costs for sampling ($O(n)$) and transformation ($O(n^2)$).

Luckily, the recursive formula for covariance shown in \eqref{eqn: corr_recursive} enables us to produce the required noise in $\Theta(n)$ time. We call our method the Cascade Sampling algorithm (Algorithm~\ref{alg:noise-mechanism}) -- it begins by sampling the noise at the highest level and propagates to the bottom most leaves, ensuring that the relationship $X_{\star} = X_{\star 0} + X_{\star 1}$  as well as the covariance matrix  in \eqref{eqn: corr_recursive} are both  preserved.
A simple yet crucial observation is the following: Given i.i.d. random variables $X, Y \sim \bN(0,1)$, define 
\begin{align}\label{eqn: one-d noise}
    X_0 = \frac{1}2 X + \frac{\sqrt 3}{2} Y,\quad X_1 = \frac 12 X - \frac{\sqrt 3}{2} Y.
\end{align}
Then $X_0, X_1$ are also $\bN(0,1)$, sum up to $X$, and have correlation $-0.5$;  Figure \ref{fig:correlated_gaussian} is a visualization. 
\begin{figure}[t]
    \centering
    \includegraphics[width = 0.24\textwidth]{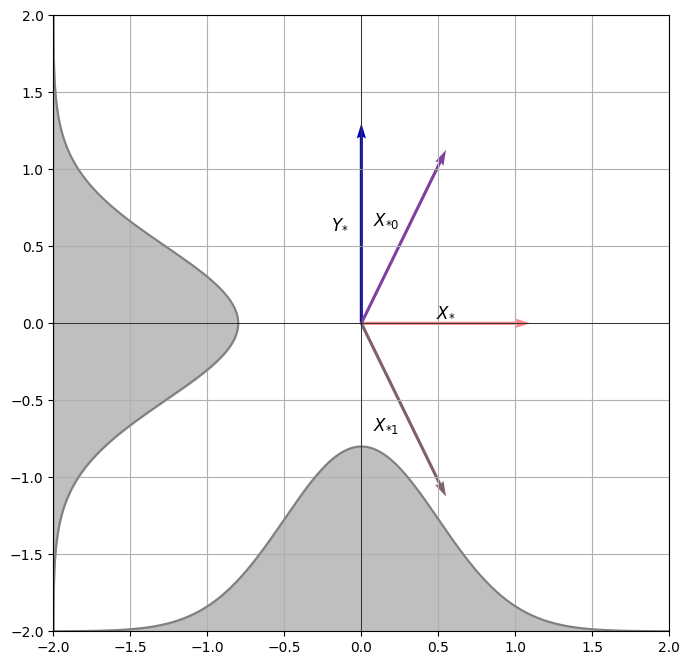}
    \includegraphics[width = 0.23\textwidth]{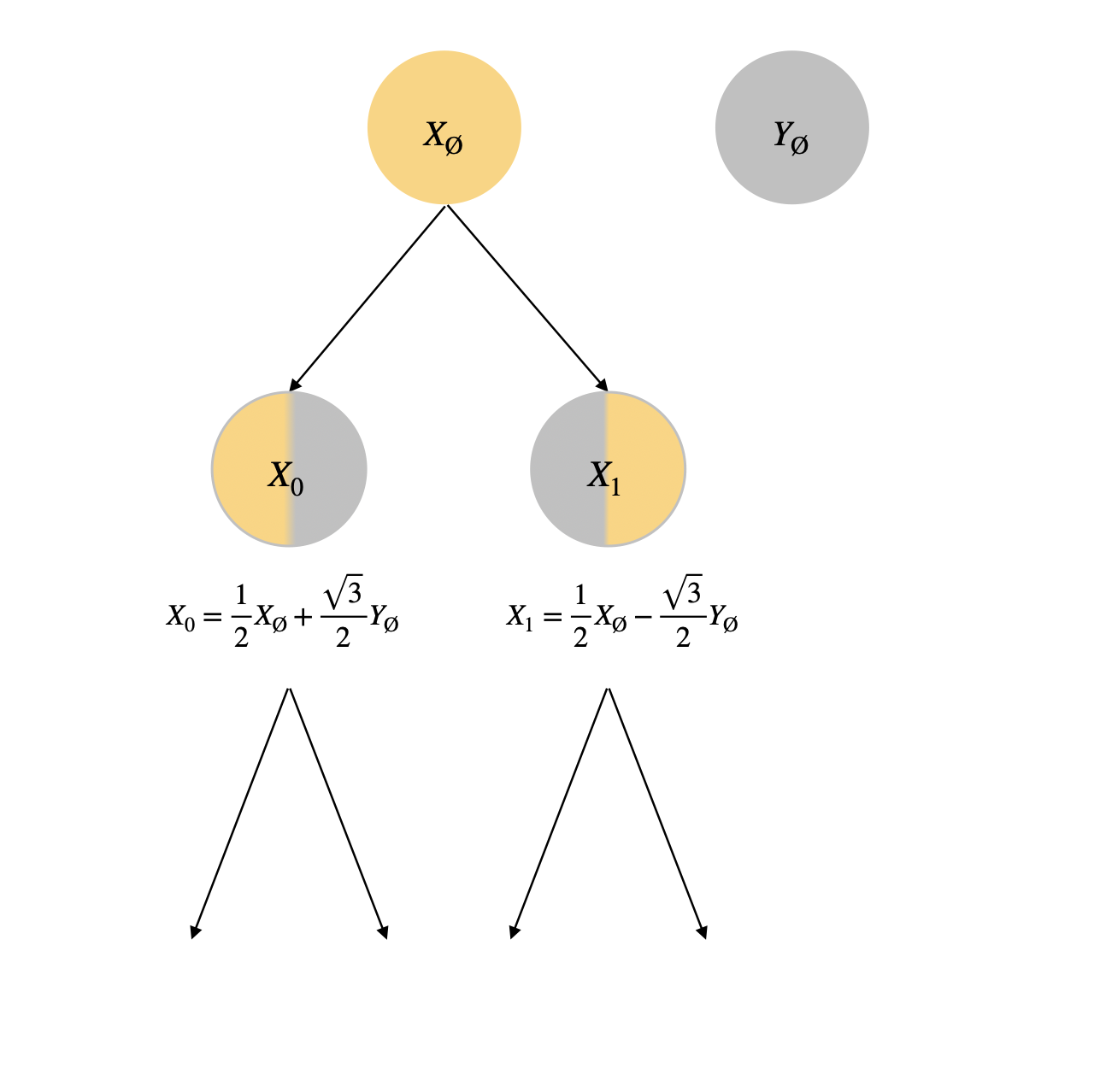}
   \caption{Left: Illustration of noise allocation to sibling nodes and their parent within a binary tree. Here $X_\star$ denotes the noise applied to a node labeled $\star$, and $Y_\star$ is another standard Gaussian  independent of $X_\star$. The noise values for the children nodes are $X_{\star 0} := X_{\star}/2 + \sqrt{3} Y_\star/2$ and $X_{\star 1} := X_{\star}/2 - \sqrt{3} Y_\star/2$. Right: Noise allocation on the top two levels of a binary tree.}
    \label{fig:correlated_gaussian}
\end{figure}
Initially, a standard normal  is sampled for the root's noise, then the noise for its direct descendants (depth 1 nodes) is determined using equation \eqref{eqn: one-d noise}. This process is repeated, applying equation \eqref{eqn: one-d noise} for each subsequent level, to assign noise to all nodes. The complete process is outlined in Algorithm \ref{alg:noise-mechanism}.

\begin{algorithm2e}
\caption{Cascade Sampling Algorithm}\label{alg:noise-mechanism}
\KwData{Depth of the binary tree $k$, variance $\sigma^2$ determined by the privacy budget.}
\KwResult{Noise values $\{X_I\}$ for all nodes $I \in \cup_{0\leq i \leq k} \{0,1\}^i$ in a binary tree.}
\For{each node $\star \in \{0,1\}^i$ at depth $0\leq i \leq k-1$}{
    \If{$\star =\varnothing$}{
    Assign $X_\varnothing \sim \bN(0,\sigma^2)$
    }
    Sample $Y_\star \sim \bN(0,\sigma^2)$ independently\;
    Define the noise values for the children of $\star$:\;
    $X_{\star 0} := \frac{1}{2} X_{\star} + \frac{\sqrt{3}}{2} Y_\star, \quad X_{\star 1} := \frac{1}{2} X_{\star} - \frac{\sqrt{3}}{2} Y_\star$\;
    }
\end{algorithm2e}

If one unit of cost is attributed to the sampling of a univariate normal variable and to each arithmetic operation (including addition or multiplication), it becomes clear that the cost of Algorithm \ref{alg:noise-mechanism} is linear in the total number of data points, dramatically improving the $O(n^3)$ cost of the standard sampling algorithm using Cholesky decomposition. 
\begin{prop}\label{prop:cost-analysis}
    Executing Algorithm \ref{alg:noise-mechanism} for a given input depth \(k\) incurs a cost of \(\Theta(n)\),
    where \(n = 2^k\) is the total count of data points (leaf nodes).
\end{prop}

Our next result shows that the leaves have the covariance structure described in Definition \ref{def:correlated-noise-mechanism}.
\begin{prop}\label{prop:noise-leaf-covariance}
   For any given positive integer \( k \), the covariance matrix of the leaf noises produced by Algorithm \ref{alg:noise-mechanism} is equal to \( \sigma^2 \bsC_k \), where \( \bsC_k \) is defined in Definition \ref{def:correlated-noise-correlation}.
\end{prop}

Finally, combining Proposition \ref{prop:noise-leaf-covariance} and Theorem \ref{thm:noise-equal-variance} immediately shows that the noises on every node  generated by Algorithm \ref{alg:noise-mechanism} have the same distribution $\bN(0, \sigma^2)$. More importantly, our correlation matrix \(\bsC_k\) has many more interesting properties that are central to privacy and utility analysis. 

\begin{remark}
    A special case of 1D range query considered substantially in the literature is continual counting in streaming data~\cite{chan2011private,FichtenbergerHU23,dvijotham2024efficient}. The ranges are of the form $[1, i]$ which reports the count (or sum) of values from index $1$ to $i$. We could adapt the cascade sampling algorithm to an incremental version to handle the streaming data. Specifically, suppose we have already received $n=2^k$ elements, and the correlated noises for the top $n$ elements have been calculated. We consider how to generate the noise for the $n+1$-th element. This will also grow the binary tree to extend for a new root $X'$, with the left child $X_0$ as the current root of the binary tree on the top $n$ elements, and the right child $X_1$ to be the root of upcoming $n$ elements. Different from the top down implementation of cascade sampling algorithm, here we have already constructed and sampled $X_0$, and we need to sample $X_1$ with negative correlation with $X_0$. This can be done by sampling $Y_\star \sim \bN(0,\sigma^2)$ independently and set $X_{1} := -\frac{1}{2} X_{0} + \frac{\sqrt{3}}{2} Y_\star$, and $X' := X_0 + X_1$. More on sampling from a conditional Gaussian distribution can be found in~\cite{gut09prob}. 
    Notice that we do not need to generate the entire subtree of $X_1$ but only need the noises along the path from $X_1$ to the $n+1$-th element. Thus the running time is $O(\log n)$ per element at most and $O(1)$ amortized. Details can be found in \Cref{subsubsec:streaming}.
\end{remark}

\begin{remark}
The work by \cite{XiaoDWZK21} considers finding the covariance matrix through an optimization procedure while constraining the variance of each workload query. As the optimization could be expensive in practice, they use approximation methods to find the desired covariance structure. For the case of equal variance on queries corresponding to nodes of a binary tree, our covariance matrix construction is a feasible solution to the optimization problem (prioritizing for privacy budget) in~\cite{Xiao2021-vi}. Our work can be seen as complementary to theirs, as our motivation comes from Census Bureaus application and identifying the recursive structure within the covariance matrix, we are able to provide a linear time sampling algorithm for this particular covariance matrix.
\end{remark}

\subsection{Privacy property}\label{subsec:privacy property}
We now turn our focus to the privacy aspects of our algorithm, with a special emphasis on identifying the suitable level of noise.  First, Theorem \ref{thm: privacy gaussian} is applicable to Gaussian noise with any covariance matrix $\bsC$. This can be seen as an extension of the traditional Gaussian mechanism (e.g. Appendix A of \cite{dwork2014algorithmic}) but for correlated noises. Then, Theorem \ref{thm:privacy} is specifically tailored for the covariance matrix outlined in Definition \ref{def:correlated-noise-correlation}.

\begin{theorem}\label{thm: privacy gaussian}
Let $X\in \mathcal{X}^n$ be any dataset, for neighboring $X$ and $X'$ let $\cM_{\sigma}(X)  = X + \mathsf{Noise}$ be the privacy mechanism, where $\mathsf{Noise} \sim \bN(\mathbf{0}, \sigma^2 \bsC)$ and $\bsC$ is an arbitrary covariance matrix with dimension $n\times n$. Fix any $\eps \in (0,1]$ and $\delta \in (0,1/2]$, the mechanism $\cM_{\sigma}(.)$ is $(\eps,\delta)$-DP for
\[\sigma^2 \geq \frac{2 \lVert \mathsf{diag}(\bsC^{-1}) \rVert_\infty \log(2/\delta)}{\eps^2}\]
where  $\lVert \mathsf{diag}(\bsC^{-1}) \rVert_\infty$  denotes the  largest magnitude of the diagonal entries of $\bsC^{-1}$. 
\end{theorem}

Theorem \ref{thm: privacy gaussian} is broad in scope yet challenging to apply. It requires the precision matrix (inverse of the covariance matrix), which is hard to estimate unless specifically designed. Fortunately, our proposed noise model has a clearly defined inverse matrix, simplifying analysis and making the theorem more practical. Our dataset includes 
$n = 2^k$  points, using multivariate Gaussian noise with a covariance of \(\sigma^2 \bsC_k\),
  as detailed in Definition \ref{def:correlated-noise-correlation}. The key findings are presented in the following theorem.

\begin{theorem} \label{thm:privacy}
	Let $X\in \mathcal{X}^n$ be any dataset and let $\cM_{\sigma}(X)  = X + \mathsf{Noise}$ be our privacy mechanism, where $\mathsf{Noise} \sim \bN(\mathbf{0}, \sigma^2 \bsC_k)$ with $\bsC_k$ defined in Definition \ref{def:correlated-noise-correlation}. Fix any $\eps \in (0,1]$ and $\delta \in (0,1/2]$, our mechanism $\cM_{\sigma}(.)$ is $(\eps,\delta)$-DP for
	\[\sigma^2 \geq \left(\frac{2}{\eps^2} + \frac{2\log_2(n)}{3\eps^2}\right)\log{\frac{2}{\delta}}  = \Theta\left(\frac{\log n\log{\frac{2}{\delta}}}{\eps^2}\right).\]	
\end{theorem}

Theorem \ref{thm:privacy} is immediate by applying Theorem \ref{thm: privacy gaussian} and Corollary \ref{cor:inverse diagonal}.

\subsection{Utility Analysis}\label{subsec:utility}
We now evaluate the utility
of our correlated input perturbation mechanism,  \(\cM_{\sigma}(\bsx) = \bsx + \bN(\mathbf{0}, \sigma^2 \bsC_k)\), for a dataset \(\bsx \in \mathcal X^n\) where \(n = 2^k\).  Given a workload matrix \(\bsW\), our mechanism operates by applying this matrix to the perturbed dataset, which we represent as \(\calW_\sigma(\bsx) := \bsW \cM_{\sigma}(\bsx)\).   
This design satisfies all the desired properties described in Section \ref{sec:intro}. First, as an additive mechanism, $\calW_\sigma$ enjoys {unbiasedness}  as the privacy noise $\bse$ has zero mean as guaranteed by design. Second, $\calW_\sigma$ maintains {consistency} as $\cM_{\sigma}(\bsx)$ could be viewed as a potential legitimate input. Furthermore, the additive construction of $\calW_\sigma$, 
coupled with the public knowledge of the noise distribution as a correlated Gaussian makes  $\calW_\sigma$  statistically transparent.

We outline various error metrics to quantify the discrepancy between $\calW_\sigma(\bsx)$ and $\bsW \bsx$. Since $\calW_\sigma(\bsx) - \bsW \bsx = \bsW \cdot \mathsf{Noise}$ where $\mathsf{Noise}\sim \bN(\mathbf{0}, \sigma^2 \bsC_k)$, we observe that the difference is a random vector that does not depend on the dataset $\bsx$. Some reasonable error metrics are as follows:
\begin{definition}[Expected total squared error]\label{def:exp-square-error} The  expected total squared error is defined as 
\begin{align*}
    \err_{\bsW, 2}(\calW_\sigma) &:= \sup_{\bsx \in \calX^n} \bE \big[ \Vert \calW_\sigma\left(\bsx\right) - \bsW \bsx \Vert_2^2 \big] \\
    &= \bE_{\bss\sim \bN(\mathbf{0}, \sigma^2 \bsC_k)} \big[ \Vert \bsW \bss \Vert_2^2 \big].
\end{align*}
\end{definition}

\begin{definition}[Worst-case expected error]\label{def:worst-case-expected} The worst-case expected error is defined as 
\begin{align*}
 \err_{\bsW}^{\infty}(\calW_\sigma) &:= \sup_{\bsx \in \calX^n} \big\Vert\bE \left[  \lvert \calW_\sigma\left(\bsx\right) - \bsW \bsx  \rvert \right]\big\Vert_\infty \\
 &= \big\Vert\bE_{\bss\sim \bN(\mathbf{0}, \sigma^2 \bsC_k)} [  \lvert \bsW \bss \rvert  ]\big\Vert_\infty.
\end{align*}
where \( |v| \) applies to each component of  \( v \).
\end{definition}

\begin{definition}[Expected worst-case  error]\label{def:expected-worst-case} The expected worst-case  error  is defined as 
\begin{align*}
 \err_{\bsW,\infty}(\calW_\sigma) &:= \sup_{\bsx \in \calX^n} \bE \left[  \Vert\calW_\sigma\left(\bsx\right) - \bsW \bsx   \Vert_\infty\right]\\
 &= \bE_{\bss\sim \bN(\mathbf{0}, \sigma^2 \bsC_k)} [ \Vert \bsW \bss  \Vert_\infty].
\end{align*}
\end{definition}

The difference between $\err_{\bsW}^{\infty}$ and $\err_{\bsW,\infty}$ arises solely from the order in which $\bE$ and the $\ell_\infty$ norm are taken.     
We have the following relationship between these errors. 
\begin{prop}\label{prop:relationship-errors}
    For any given $\sigma >0$ and query matrix $W$ of size $m\times n$, we have:
    \begin{align}
        \sqrt{ \err_{\bsW, 2}(\calW_\sigma)/m}\leq  \err_{\bsW}^{\infty}(\calW_\sigma) \leq \err_{\bsW,\infty}(\calW_\sigma). \label{eqn:relationship-errors}
    \end{align}
\end{prop}

The focus of our analysis is on two scenarios. The first is that \(\bsW\) represents all consecutive range queries. In this context, \(\bsW\) is a binary matrix of dimensions \((\binom{n}{2} +n) \times n\), with each row comprising a sequence of consecutive ones. The second is when  \(\bsW\) represents all `nodal' queries, indicating that \(\bsW\) is a matrix of dimensions \((2n-1) \times n\), designed to query the values associated with every node in our binary tree.
We also recall the notations introduced in Section \ref{subsec:correlation-matrix}, where $\{X_I\}_{I\in\{0,1\}^k}$ represents the Gaussian noise vector with a covariance matrix $\sigma^2 \bsC_k$.

\subsubsection{Continuous range queries}
When $\bsW$ encompasses all continuous range queries, the related noise $\calW_\sigma(\bsx)-\bsW \bsx$ forms a vector of length $\binom{n}{2} + n$. Each element in this vector represents a consecutive sum $\sum_{L = I}^{I+j} X_L$. It is evident that each element is a univariate Gaussian with zero mean. The essential technical lemma  below demonstrates that the maximum variance increases logarithmically with the number of data points.

\begin{lemma}\label{lem:maximum-consecutive-variance}
Let $\{X_I\}_{I\in\{0,1\}^k} \sim \bN(\mathbf{0}, \sigma^2 \bsC_k)$ with $\bsC_k$ defined in Definition \ref{def:correlated-noise-correlation}. Then the maximum variance among all the consecutive sums satisfies:
 \begin{align*}
    \max_{I \in \{0,1\}^k, j\leq 2^k - I}\var[\sum_{L = I}^{I+j} X_L]/\sigma^2 = \Theta(k) = \Theta(\log_2(n)).
\end{align*}   
\end{lemma}

With Lemma \ref{lem:maximum-consecutive-variance}, we are ready to state the utilities of our privacy mechanism under different metrics.
\begin{theorem}\label{thm:utility-continuous}
     Fix $\sigma >0$ and let query matrix $\bsW$ be all the continuous range queries, we have:
     \begin{align*}
     \err_{\bsW}^{\infty}(\calW_\sigma) = \Theta\left(\sigma \sqrt{\log_2(n)}\right), \\
     \err_{\bsW,2}(\calW_\sigma) = O\left(\sigma^2 n^2 \log_2(n)\right), \\
     \err_{\bsW,\infty}(\calW_\sigma) = O(\sigma\log_2(n)).
     \end{align*}
\end{theorem}

Choosing \(\sigma\) to accord to the privacy budget, the next corollary gives our mechanism's utility guarantee. 
\begin{corollary}\label{cor:utility-continuous-privacy}
    Let $X\in \mathcal{X}^n$ be any dataset and let $\cM_{\sigma}(X)  = X + \sigma\bN(\mathbf{0}, \bsC_K)$ be our privacy mechanism, where $\sigma$ is chosen such that the mechanism satisfies $(\eps,\delta)$-DP. Let $\bsW$ be all the continuous range queries, we have:
    \hspace{-0.5pt}
    \begin{align*}
        \err_{\bsW}^{\infty}(\calW_\sigma) &= \Theta\left( \log(n) \sqrt{\log(2/\delta)} \eps^{-1}\right), \\
        \err_{\bsW,2}(\calW_\sigma) &= O\left(n^2 \log^2(n)\log(2/\delta)\eps^{-2}\right), \\
     \err_{\bsW,\infty}(\calW_\sigma) & = O\left( \log^{1.5}(n) \sqrt{\log(2/\delta)} \eps^{-1}\right).
   \end{align*}
   
\end{corollary}

Theorem \ref{thm:utility-continuous} differs from the scenario where independent Gaussian noise \(\bN(0, \sigma^2)\) is added to each leaf node. In the latter case, the three errors, $\err_{\bsW,2}, \err_{\bsW}^{\infty},$ and $\err_{\bsW,\infty}$, are \(\Theta(n^3 \sigma^2)\), \(\Theta(\sqrt{n}\sigma)\), and \(\Omega(\sqrt{n}\sigma)\) respectively. In contrast, our mechanism results in errors of \(\Theta(n^2\log(n)\sigma^2)\), \(O(\sqrt{\log(n)}\sigma)\), and \(O(\log(n)\sigma)\), indicating lower error magnitudes for all evaluated metrics.

\textbf{Comparison of bounds in \cref{cor:utility-continuous-privacy}:} 
The bounds attained by our mechanism for $\err_{\bsW,2}(\calW_\sigma)$ match previous work \cite{hay2009boosting, XiaoWG11} which can be viewed as instantiations of the matrix mechanism \cite{li2015matrix}. To  our best knowledge, this is the first work that obtains comparable bound to specialized output perturbation mechanisms for range queries. We also show that empirically we match the performance of these methods in \cref{sec:experiments}. Though these works do not explicitly analyze for $\err_{\bsW,\infty}(\calW_\sigma)$,  \cite{JRSS23} showed that the Binary Tree mechanism \cite{DworkNPR10, chan2011private} with Gaussian noise obtains the same $O(\log^{3/2}n)$ bound.

In terms of tightness, our bound for $\err_{\bsW,2}(\calW_\sigma)$ is optimal among all  $(\eps,\delta)$-DP mechanisms. To see this, note that the workload matrix for continual counting (a lower triangular matrix of size $n \times n$ with 1 for entries on  and below the diagonal and 0 elsewhere) forms a subset of the workload matrix for continuous queries. \cite{HenzingerUU23} (Thm. 4) show $\Omega(\log^2 n)$ lower bound on error for any $(\eps, \delta)$-DP mechanism for the continual counting  which matches our upper bound (after scaling for the number of queries). For $\err_{\bsW,\infty}(\calW_\sigma)$, \cite{FichtenbergerHU23} (Thm. 3) show $\Omega(\log^2 n)$ lower bound on the \textit{squared-infinity}  error for any $(\eps, \delta)$-DP mechanism, leading to an $\Omega(\log n)$ lower bound for $\err_{\bsW,\infty}(\calW_\sigma)$. Consequently, our upper bound for  $\err_{\bsW,\infty}(\calW_\sigma)$ in  \cref{cor:utility-continuous-privacy} is only off by a factor of $\log^{0.5} n$ (due to Gaussian concentration) compared to the known lower bound. An intriguing open question is whether this gap can be narrowed from either side.

\section{A brief review of related literature}\label{subsec:literature-review}

This section provides an abridged review of existing DP mechanisms for linear queries, with a focus on how they can be used for the Census application with its considerations.
Due to space constraints, we defer an extended literature review to Appendix~\ref{sec:lit-review-appendix}.

As a quintessential {\it output perturbation} method, \emph{Gaussian mechanisms} add to the query output centered Gaussian noise with variance 
tailored to its sensitivity~\citep{Dinur2003,Dwork2004-dx,Dwork2006,hardt2010geometry,nikolov2013geometry}, which for linear queries is the largest norm of the columns of workload matrix $\bsW$. 
Gaussian mechanisms are unbiased and statistically transparent, but are not flexible in utility control because the accuracy of the entire output query is dictated by $\bsW$. They are not consistent in general unless $e$ is guaranteed to be in the column space of $\bsW$, as discussed in Remark~\ref{remark:consistency}.

The Census Bureau's \emph{TopDown algorithm}~\citep{abowd2022topdown,Abowd2024-jj} is a massive endeavor to navigate through the complex requirements discussed in Section~\ref{sec:intro}. It consists of two phases: 1) the ``measurement'' phase injects additive discrete Gaussian noise \citep{canonne2020discrete} to confidential queries, similar to the Gaussian mechanisms; 2) the ``estimation'' phase uses optimization-based post-processing to achieve consistency (internal and external) and complex utility control such as invariants/consistency, external constraints, and marginal accuracy targets. Post-processing does not damage privacy protection -- which is one of the greatest traits of DP. Unfortunately, the post-processing phase of TopDown algorithm destroyed both unbiasedness~\cite{NAP25978,gao2022subspace} and statistical transparency~\cite{gong2022transparent,hawes23census} (noise distribution became intractable and regression on data output revealed unwanted bias). Researchers tend to agree that a major contributor to the bias is the nonnegativity constraint that TopDown’s post-processing enforces; see e.g.~\cite{abowd21uncertainty}. 



In general consistency could be obtained by projection into a designated subspace~\cite{li2015matrix} or other optimization-based transformations. Special care should be taken to ensure unbiasedness (e.g.~\cite{barak2007privacy,hay2009boosting}), which could be lost if inequality constraints are present. Moreover, it would be ideal that an analytical description is available, e.g.~\cite{hay2009boosting}, to maintain statistical transparency. 
In addition, though theoretically efficient algorithms are known, certain post-processing approaches are not practical in real-world use cases of very large data sets, e.g. projecting onto convex sets. For example, the post-processing phase of TopDown algorithm is a substantially non-trivial effort in terms of computational costs; see section 7.6 of~\citep{abowd2022topdown}.

To improve utility, a lot of work has been developed to carefully design how perturbations are added. Starting from the work in~\cite{li2015matrix}, the broadly defined family of \emph{matrix mechanism} is a workload-dependent mechanism. In general, one finds a privileged factorization, $\bsW = RA$, and infuses Gaussian noise to the intermediate result $Ah$, with $h$ being the histogram vector. 
Choice in the factorization allows for the attainment of near optimal utility. 
While this can be done for any general workload matrix, for the special case of one dimensional range query, there has been a number of excellent work such as Binary Tree \citep{DworkNPR10, chan2011private}, Hierarchical \citep{hay2009boosting}, Wavelet \citep{XiaoWG11}, EigenDesign \citep{LiM12}, HDMM~\cite{McKenna_Miklau_Hay_Machanavajjhala_2023} and explicit factorization~\cite{FichtenbergerHU23} mechanisms that can be considered as carefully choosing the matrix factorization. Consistency of each of these mechanisms can be individually checked.
We will review the details of these algorithms in Appendix~\ref{sec:lit-review-appendix}. We also show empirical results in the next section.

We would like to remark that our mechanism can also be written as a matrix factorization $M(x) = R(Ax + z)$ with $R = \sqrt{\bsC_k}$, $A = R^{-1}W$ where $W$ is the query workload matrix for 1D range query. In addition to showing asymptotically tight $\ell_2$ error and near-optimal $\ell_{\infty}$ error (off by a $\sqrt{\log n}$ factor from the current best lower bound), our mechanism also ensures that all queries share the \emph{same} variances.

\section{Experiments}\label{sec:experiments}




\begin{figure*}
    \centering
    \subfigure[]{\includegraphics[width=0.69\columnwidth]{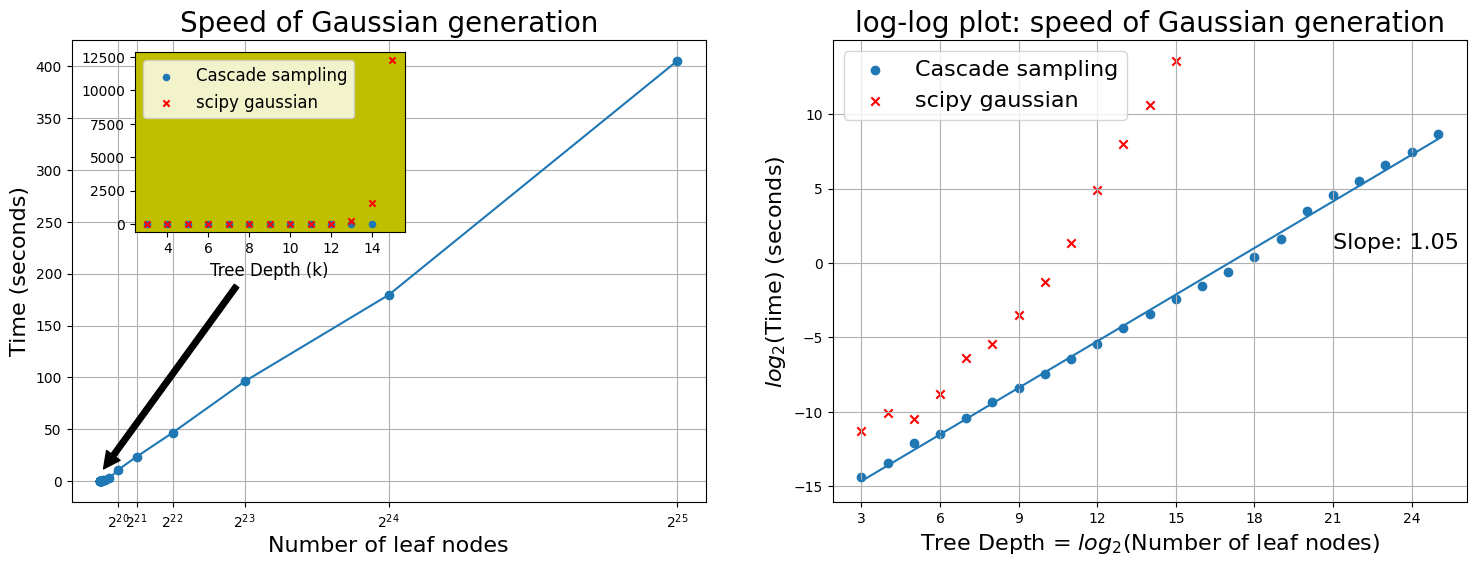}\label{fig:time-comparison}} 
    \subfigure[]{\includegraphics[width=0.32\columnwidth]{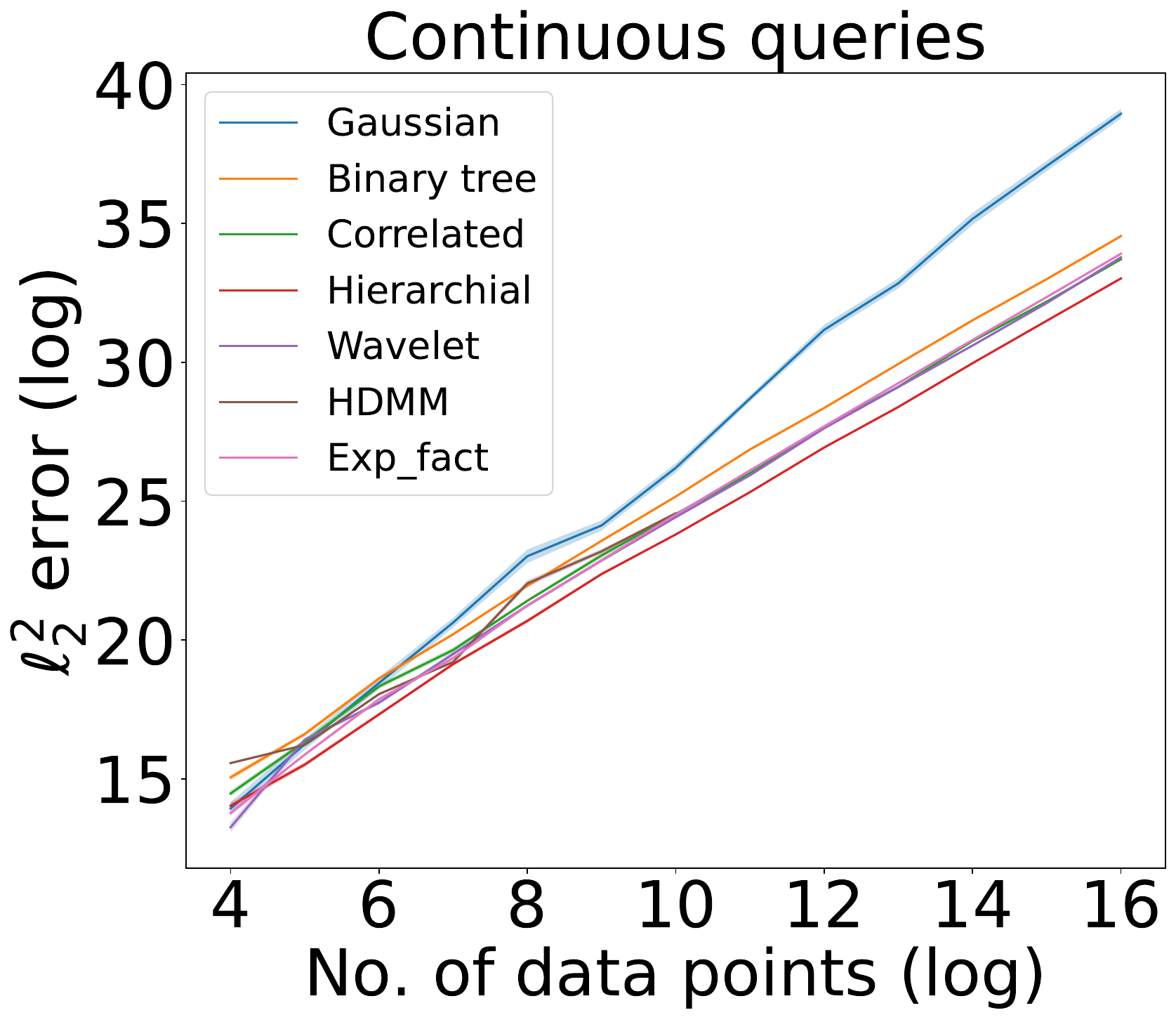}\label{fig:cont-l2}} 
    \subfigure[]{\includegraphics[width=0.32\columnwidth]{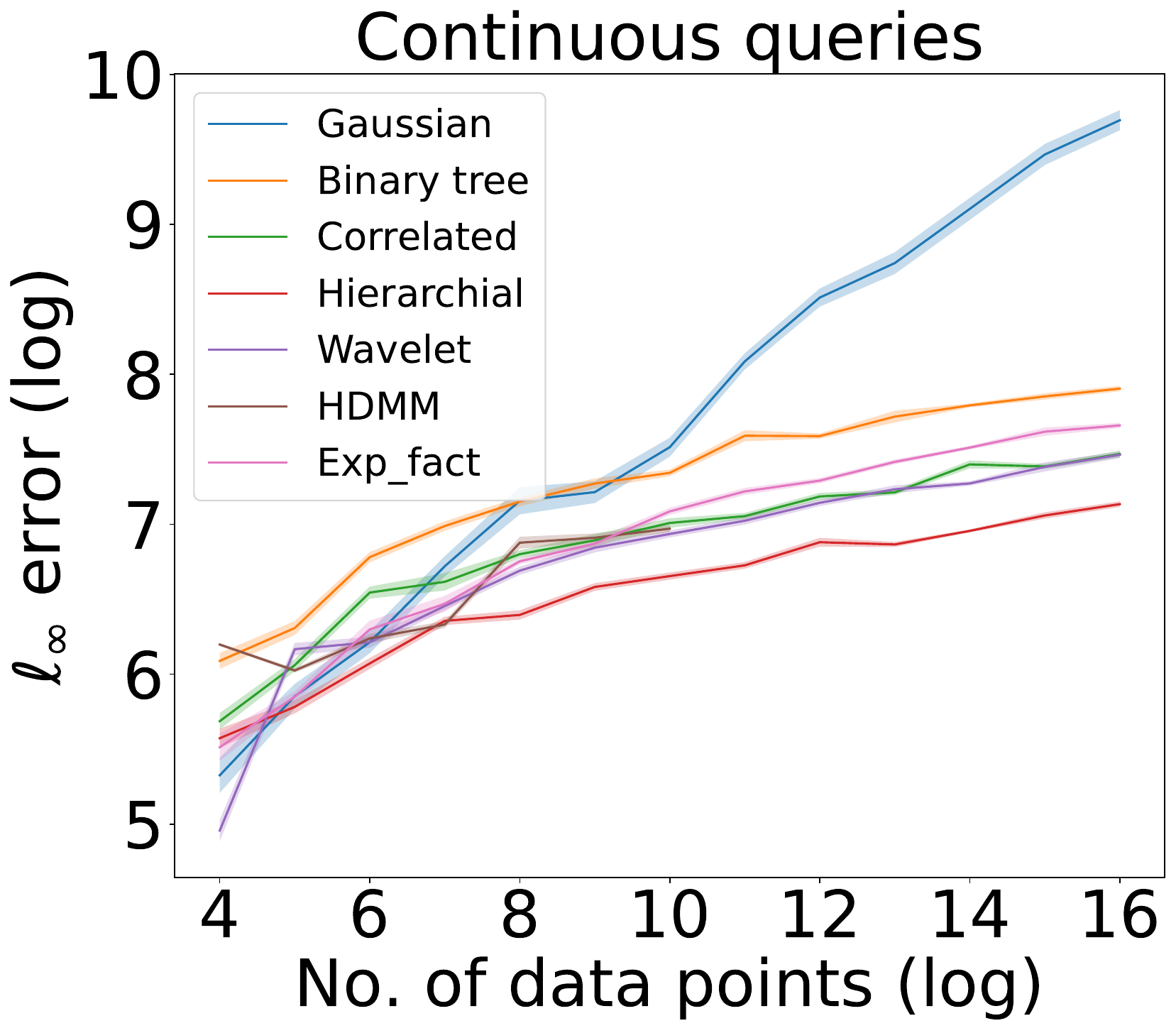}\label{fig:cont-linf}}
    \subfigure[]{\includegraphics[width=0.33\columnwidth]{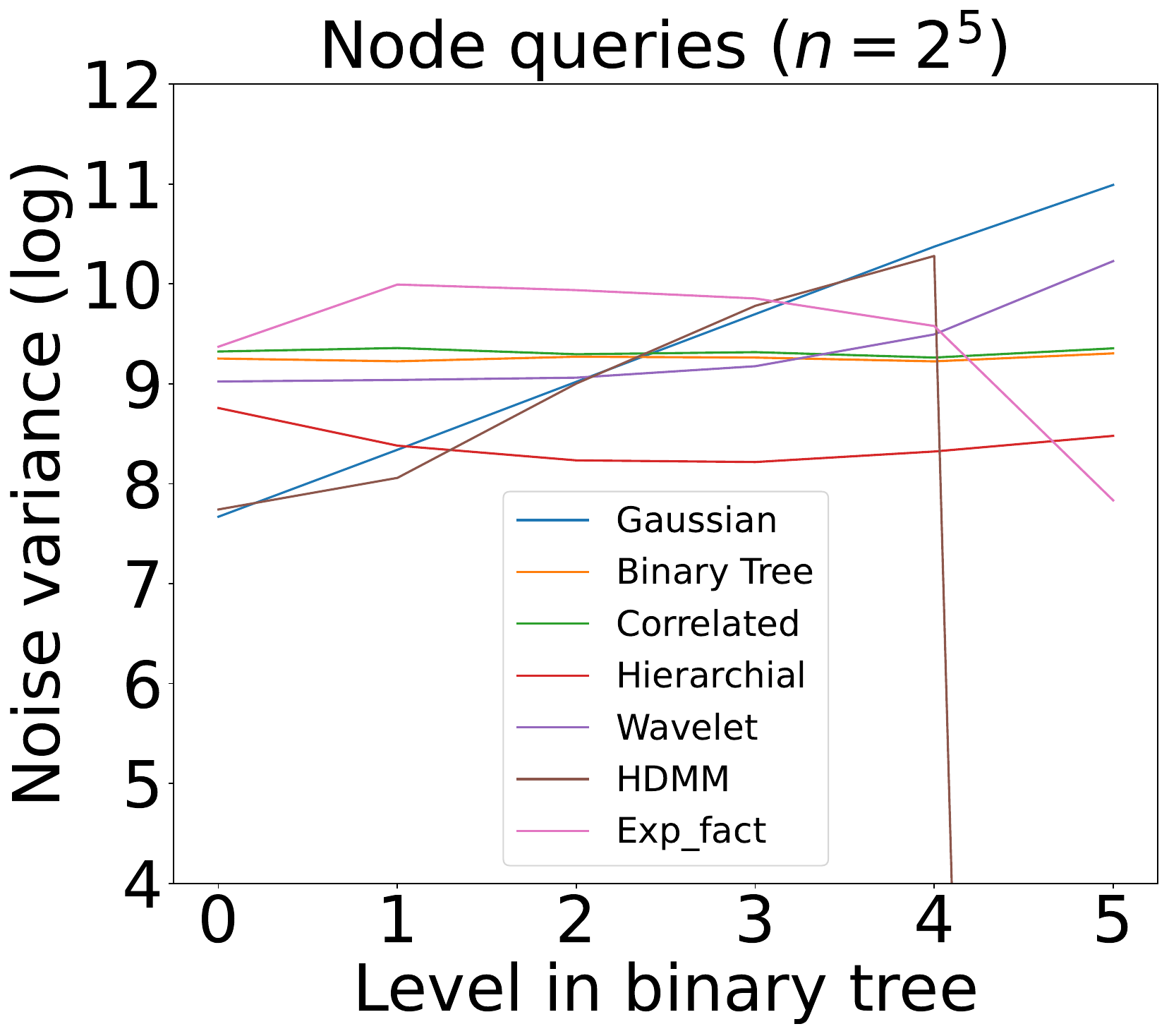}\label{fig:var-5}}
    \subfigure[]{\includegraphics[width=0.32\columnwidth]{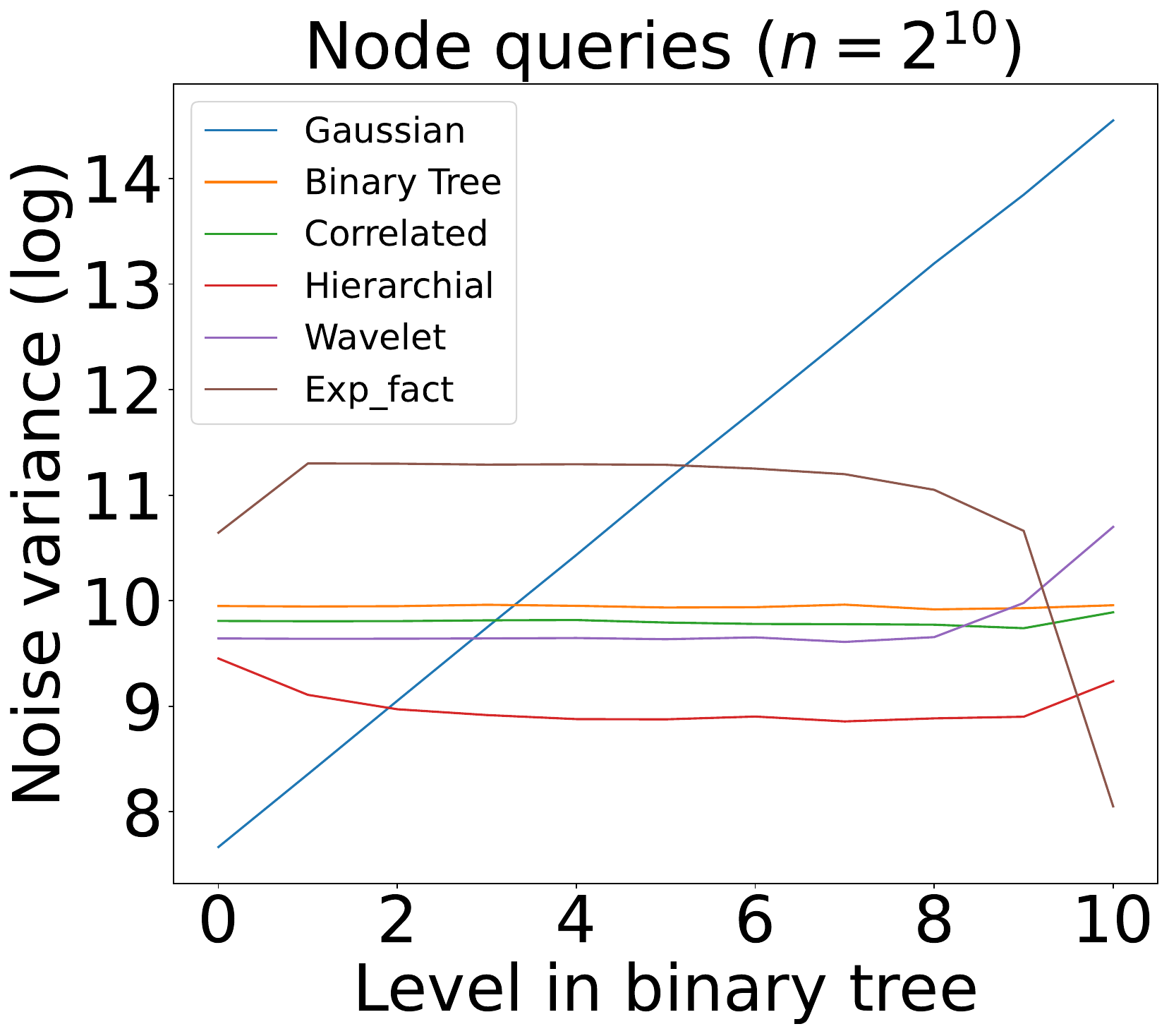}\label{fig:var-10}}
    \caption{(a): Comparison of Gaussian Generation Speeds (Left: original scale; Right:log-log scale). (b) \& (c): $\err_{\bsW,2}(\calW_\sigma)$ and $\err_{\bsW,\infty}(\calW_\sigma)$ for continuous queries in log-log scale.  
    (d) \& (e): Variance of queries for levels of binary tree in log-log scale. }
    \label{fig:results}
\end{figure*}



\textbf{Efficiency of Cascade Sampling.} The computational costs in terms of clock time for both scales, original and logarithmic are displayed in Figure \ref{fig:time-comparison}. Cascade Sampling  demonstrates remarkable efficiency, requiring less than a second for up to \(n = 2^{17} \) ($131,072$). Even for  a over $33.5$ million-dimensional Gaussian, our algorithm completes the task in under $9$ minutes. Right part of Figure \ref{fig:time-comparison} features a scatterplot of log(Time) against log(Data Points), with a fitted straight line that has a slope of $1.05$, confirming our assertion of linear cost. Additional details are provided in \cref{app:add_exp}.

\textbf{Utility comparison with existing algorithms. }
We evaluate the correlated input perturbation mechanism against the Gaussian (adding noise to $\bsx$ and querying from the privatized data), Binary Tree \citep{DworkNPR10, chan2011private}, Hierarchical \citep{hay2009boosting}, Wavelet \citep{XiaoWG11}, HDMM~\cite{McKenna_Miklau_Hay_Machanavajjhala_2023} and explicit factorization~\cite{FichtenbergerHU23}.  We do not compare with an output Gaussian mechanism (adding noise to $\bsW\bsx$) as this performs substantially worse due to high sensitivity. Experiments comparing runtime are in \cref{app:exp_runtime}. 

\underline{Experimental setup:} We compare mechanisms on randomly generated data, where each element of $\bsx \in \mathbb{Z}^n$ is sampled uniformly at random from $1$ to $1000$. We consider the case where a person contributes a single count to an element in $\bsx$ and hence for neighboring $\bsx$ and $\bsx'$, $\|\bsx - \bsx'\|_1 = 1$.
We consider $n$ in range $\{2^{4}, \cdots , 2^{15}\}$, except for HDMM which we restrict to $2^{10}$ due to computational overhead. 
We fix privacy parameters $\eps = 0.1$ and $\delta = 10^{-9}$ for all our experiments and report $\err_{\bsW,2}(\calW_\sigma)$ and $\err_{\bsW,\infty}(\calW_\sigma)$ (\cref{def:exp-square-error} and \cref{def:expected-worst-case}) for different choices of workload matrices $\bsW$. The error values on graphs are plotted with $\pm 0.25$ standard deviation across 10 independent runs for all experiments. Experiments for node and random queries and additional implementation details are included in \Cref{app:add_exp}.

\underline{Continuous range queries:}
\label{subsec:exp_cont} 
A continuous range query is specified by start and end index in $\{1,\ldots, n\}$ and the answer is sum of all elements within this range. We consider the workload matrix of all possible $\binom{n}{2} + n$ queries. As computing output of all these queries becomes computationally infeasible for very large $n$, we  randomly sample the queries to obtain a Monte Carlo estimate of the errors, with details explained in Appendix \ref{subsec:exp_details}. We sample $5000$ queries for each value of $n$ and report estimated $\err_{\bsW,2}(\calW_\sigma)$ and $\err_{\bsW,\infty}(\calW_\sigma)$ in \cref{fig:cont-l2,fig:cont-linf}.

    

Both $\err_{\bsW,2}(\calW_\sigma)$ and $\err_{\bsW,\infty}(\calW_\sigma)$ follow the same trend. 
Four mechanisms, ours, Wavelet, HDMM and Hierarchical, appear to be on the first tier (i.e., with smallest error) for both $\ell_2$ and $\ell_{\infty}$. 
For small $n$, Gaussian mechanism performs as well as these methods but suffers variance scaling linearly in $n$ for larger values. Our mechanism, owing to its design, exhibits significantly reduced error compared to the standard Gaussian approach, its primary input perturbation competitor. Our method
achieves comparable error empirically with the state of the art methods,
yet it carries additional nice statistical properties such as unbiasedness, consistency, and transparency. 


\underline{Utility control:} The ability to control noise variance across different `levels' of summation was one of the motivations of our mechanism, which we demonstrate here. We consider a binary tree over the dataset vector $\bsx$, and the queries (which are nodes of the tree) correspond to intervals summing up the leaves of the sub-tree rooted at the node. We then consider the variance of noise across queries corresponding to different levels (and the root would correspond to the highest level summing up all the elements of $\bsx$). We perform this for $n = 2^5$ and $n=2^{10}$ for $500$ independent runs and the results are shown in \cref{fig:var-5,fig:var-10}.

    

As the additive correlated noise in our mechanism is designed to have the same variance across levels, our mechanism and the Binary Tree mechanism, which adds i.i.d. Gaussian noises to each element of binary tree to answer queries have the same variance across all such queries. Comparatively, Gaussian mechanism has variance linearly increasing across levels. Wavelet has smaller consistent variance for smaller levels (corresponding to short ranges) but increases sharply for long ranges, whereas the opposite effect is observed for explicit factorization and HDMM. Hierarchical has the smallest variance in magnitude and also a typical trend --- variance is larger for smaller and larger levels and less for intermediate ones, which helps illustrate the effect of post-processing and explains the improved utility.

\section{Discussion and Future Work}
\label{sec:discussion}

The proposed input perturbation mechanism leverages error correlation structures which, once combined with the query workload matrix,  allows for fine-grained control over error magnitude across queries at different geographic levels, a feature that is highly desirable in applications such as the dissemination
of confidential official statistical data products with complex structures.
Our proposal permits straightforward generalizations to general (unbalanced) binary trees, as well as higher-dimensional query types, in particular 2-D range queries and contingency tables that recognize, for example, the geographic contiguity of the data points. Both generalizations need to modify the noise allocation scheme in the cascade sampling algorithm. Appendices~\ref{subsubsec:general_binary_tree} and~\ref{subsubsec:2d-queries} respectively discuss how to instantiate these generalizations. 
We believe this work leads to a generic framework of differential privacy mechanism design.  In particular, beyond equality of error variance at every hierarchical level considered here, the Cascading Sampling Algorithm lends itself to a straightforward generalization where detailed variance control is required, i.e. when different error variances are needed at different levels.  Furthermore, when the data or query workload structure is well understood, additional benefits can result from catering privacy perturbations to respect these known structures. These investigations are left to future work. 


\newpage

\bibliographystyle{abbrv}
\bibliography{references}


\newpage
\onecolumn
\appendix

\section{An extended review of related literature}\label{sec:lit-review-appendix}



Linear query is a central topic that has been studied extensively in the differential privacy literature. To supplement the brief review provided in Section~\ref{subsec:literature-review}, we discuss in greater detail existing work on privacy mechanisms for linear queries in relation to the properties of unbiasedness, consistency, statistical transparency, and flexible utility control.

A simple mechanism is input perturbation in which data entry either perturbed (in a discrete setting) or infused with an independent noise (often from a Laplace or Gaussian distribution) with suitable magnitude. 
The linear query can be answered by simply reporting the query answer on the perturbed data. As Section~\ref{subsec:literature-review} alludes to, such mechanisms supports unbiasedness (if the element-wise perturbation has zero mean), logical consistency, statistical transparency, and are often amenable to distributed implementation. However, it does not support flexible utility specifications. In particular, when a query range contains a large number of data elements, the query utility (or error) grows in the order of $O(\sqrt{n})$ where $n$ is the the number of elements in the range. The error blow-up for aggregated data substantially hinders the use of input perturbation in practice.



In contrast, mechanisms that use {\it output} perturbation consider adding noise to the query output. One classic example is the Gaussian mechanism~\cite{Dinur2003,Dwork2004-dx,Dwork2006}. The variance of independent Gaussian noise added to each query needs to be tailored to the query {\it sensitivity}, i.e., the maximum number of queries one element is involved, or the largest norm of the columns of workload matrix $\bsW$. Gaussian mechanisms satisfy unbiasedness and statistical transparency, but does not support flexible utility control. 
The noise magnitude could be high when the sensitivity of $\bsW$ is high. For consistency, if the perturbation noise $\bse$ is in the column space of $\bsW$, then consistency holds -- we can write the data output as $\bsW \bsx+\bse=\bsW(\bsx+\bse')$ with $\bsW\bse'=\bse$.


A notable class of Gaussian mechanisms provide improved utility to linear queries by leveraging their geometric structure \cite{hardt2010geometry,nikolov2013geometry}.
Specifically, consider the workload matrix $W$ as a linear map that maps the input data histogram to a $d$-dimensional output vector, with $d\ll n$. \cite{hardt2010geometry} studied the image of the unit $\ell_1$ ball under the linear map $\bsW$ (which is a convex polytope $K$) and use an instantiation of the exponential mechanism (via a randomly chosen point in $K$) for $\eps$-differential privacy. In a later work, \cite{nikolov2013geometry} considered the setting when the number of queries $d$ is much larger than $n$, and $(\eps,\delta)$-differential privacy in general. They use a Gaussian mechanism where the parameters are guided by the minimum enclosing ellipsoid of a square submatrix of $\bsW$. 
These mechanisms have polylogarithmic approximation (in terms of error) to the lower bound for such queries, which are derived by using hereditary discrepancy of $\bsW$. As special cases of Gaussian mechanisms, these mechanisms also fall in the category of output perturbation, and enjoy similar properties as does the classic Gaussian mechanism in terms of bias, consistency, transparency, and utility control. 


The generic family of matrix mechanism~\cite{li2015matrix,McKenna_Miklau_Hay_Machanavajjhala_2023}, also called the factorization mechanism~\cite{Nikolov16geometry}, is workload-dependent (i.e. the workload matrix $\bsW$ plays a role in its design) and is situated between input and output perturbation.
The workload matrix $\bsW$ is factored into two matrices $R$ and $A$, with $R$ called the reconstruction matrix and $A$ the strategy matrix.  Gaussian errors are added to the intermediate result $Ah$, where $h$ is the histogram vector. The final query result is taken as the multiplication of $R$ and the noisy intermediate result. This way the sensitivity can be controlled by the maximum column norm of $A$. Depending on the utility objective (e.g., $\ell_{\infty}$ or $\ell_2$ of the final error), one can choose the factorization $\bsW=RA$ and final query error is proportional to the multiplication of the max row norm of $R$ \footnote{To optimize for the final $\ell_{\infty}$ (or $\ell_2$) error one would minimize the max $\ell_2$ row norm (or Frobenius norm) of $R$ and maximum $\ell_2$ column norm of $A$~\cite{li2015matrix,Nikolov16geometry}.} and maximum $\ell_2$ column norm of $A$.
By choosing an appropriate factorization $\bsW = RA$ (e.g., by using convex optimization~\cite{Linial2009-tn} for final $\ell_{\infty}$ error), one can get near optimal utility 
\cite{nikolov2013geometry,Edmonds2020-qq}. 
The matrix mechanism is unbiased, has statistical transparency and consistency if the noise is in the column space of $A$. Utility control depends on the specific choice of the matrix factorization. For efficient implementation, \cite{McKenna_Miklau_Hay_Machanavajjhala_2023} give an efficient implementation of the matrix mechanism by considering different representations of the workload matrix thereby reducing the space for the optimization problem.


 
 A number of prior work considered range query problems and developed hierarchical methods, which can be considered as special cases of the matrix mechanism~\cite{hay2009boosting,chan2011private,QardajiYL13,XiaoWG11, honaker2015efficient}. We use
\cite{chan2011private} as an example to illustrate the main idea. \cite{chan2011private} considered a specific type of range query that is derived from continuous counting in a data stream. Specifically, consider a set of binary data points $x_i\in \{0, 1\}$ arriving sequentially, the goal is to report the count at any time $t$, i.e., $X_t=\sum_{i=1}^{t}x_i$, for all $t$, with differential privacy guarantee for individual data points. 
They developed a mechanism that renders the error utility of bound $O((1/\eps)\log^{3/2} t)$ for answering the $t$-th counting query with $\eps$-differential privacy. The main idea is to maintain a binary tree on the data stream $x_1, \cdots, x_t$ and add an independent Laplace noise for each internal node in the binary tree. The final query output will be derived by the summation of the (noisy) count of selected vertices in the binary tree. The mechanism can be applied to a general 1D range query problem of reporting $\sum_{k=i}^{j} x_k$, for any 1D range $[i, j]$, with $i\leq j$, with the same error bound. 
A number of other work use similar ideas: \cite{barak2007privacy} uses subsets of Fourier basis, \cite{XiaoWG11} uses wavelet transforms, and \cite{muthukrishnan2012optimal} considers general range queries with constant VC-dimension (e.g., half-space queries). \cite{FichtenbergerHU23} focuses on analysis of constant factors in the error bound using matrix factorization for the continual counting scenario and also provides concrete matrix choices for $R$ and $A$. There are additional work that focuses on space efficiency of implementation in the streaming setting~\cite{dvijotham2024efficient} or efficiently solving for the matrix mechanism for a wide variety of convex objective functions~\cite{xiao23advances}.

\emph{Subspace DP} mechanisms~\citep{gao2022subspace,dharangutte2023integer} impose external consistency on the data product without post-processing. External consistency is expressed through a set of invariants, exact queries from the confidential database, and may be understood as a special utility requirement of zero privacy noise.
Subspace DP mechanisms allow utility control over all invariant queries, unbiasedness, and statistical transparency, but do not generally guarantee internal consistency.

We make a brief comment on the practicality of implementation. Among the mechanisms discussed here and in Section~\ref{subsec:literature-review}, some are mainly of theoretical interest~\cite{Hardt2012-sh,nikolov2013geometry,muthukrishnan2012optimal}. The matrix mechanism has a practical instantiation~\cite{mckenna2018disclosure}, as does~\cite{hay2009boosting} and the subspace DP mechanisms of~\cite{gao2022subspace,dharangutte2023integer}. In addition, the TopDown algorithm has been extensively instantiated, though it is difficult to independently replicate due to its formidable scale.

Finally, we remark that the local differential privacy model (LDP) refers to the scenario when each individual generates local noises and shares the perturbed value, without any need of a single trusted server. Although our noises are applied locally for each data value and the perturbed value is shared publicly, however, the generation of the local noises would require secured communication with either a trusted server (to run the cascade sampling algorithm) or with other peers/individuals (to exchange random seeds that they both use to help to install the negative correlation in their generated noises).

\section{Technical proofs}\label{appendix:proofs}

We first prove Theorem \ref{thm:noise-equal-variance}.  
Given that $\sigma$ merely functions as a multiplicative constant within our privacy mechanism, our attention will now shift towards understanding some fundamental characteristics of $\bsC_k$, as listed below:

\begin{prop}\label{prop: correlation-basic-properties}
For a fixed positive integer $k$, the matrix $\bsC_k$ exhibits the following properties:
    \begin{enumerate}
    \itemsep0em 
         \item The diagonal entries of $\bsC_k$ are all $1$.
        \item Every off-diagonal element of $\bsC_k$ is negative.
        \item \label{property:summation}Summing over every column $j$ for any row $i$, we have $\sum_{j} \bsC_k(i,j) = 2^{-k}$.
        \item $\bsC_k$ is a positive semi-definite (PSD) matrix.
    \end{enumerate}
\end{prop}
\begin{proof}[Proof of Proporition \ref{prop: correlation-basic-properties}]\label{prof:correlation-basic-properties}
    The initial two properties can be readily validated based on the matrix's definition. For the third property, it is evidently correct when $k = 1$. Assuming its correctness for $k = l$, when $k= l+1$, we can deduce from \eqref{eqn: corr_recursive} that
    \begin{align*}
            \sum_j \bsC_{l+1}(i,j) &= \sum_j \bsC_l(i,j) + (-2^{-2l-1}) \times 2^l  \\
            & = 2^{-l} - 2^{-l-1} = 2^{-l-1} = 2^{-(l+1)},
    \end{align*}
as desired. The final property can be directly applying the Gershgorin circle theorem based on the first three properties.
\end{proof}

\begin{proof}[Proof of Theorem \ref{thm:noise-equal-variance}]\label{prof:noise-equal-variance}
   Assuming, without a loss of generality, that $\sigma = 1$, we begin our proof by considering the case where \( k = 1 \). This essentially requires us to demonstrate that \( X_\varnothing := X_0 + X_1 \) has a unit variance. We have:
   \begin{align*}
       \var[X_\varnothing] = (1,1) \cdot \bsC_1 \cdot (1,1)^\top = 1 + 1 - 0.5 - 0.5 = 1,
   \end{align*}
   as desired. Now, let's assume the aforementioned result holds true for \( k = l \). For \( k = l+1 \), the formula \eqref{eqn: corr_recursive}, combined with our inductive assumption, confirms that every node in both the left and right subtrees maintains the same distribution \( \bN(0,1) \). Consequently, our task narrows down to computing the variance for the root node \( X_\varnothing \), which is given by:
    \begin{align*}
       \var[X_\varnothing] &= (1,1, \ldots, 1) \cdot \bsC_{l+1} \cdot (1,1,\ldots, 1)^\top \\
       & = \sum_{i,j} \bsC_{l+1}(i,j) \\
       & = \sum_{i} \sum_{j} \bsC_{l+1}(i,j)\\
       & = \sum_{i} 2^{-(l+1)} = 1 \qquad \text{by the Property \ref{property:summation} in Proposition \ref{prop: correlation-basic-properties}.}
   \end{align*}
   This concludes our proof. 
\end{proof}

\begin{proof}[Proof of Proposition \ref{prop:noise-leaf-covariance}]
   Without loss of generality, we assume $\sigma = 1$. We will prove by induction. When $k = 1$, it suffices to show the  $\var(X_0) = \var(X_1)  = 1$ and $\cov(X_0, X_1) = -0.5$. For the former, observe that 
   \begin{align*}
        \var(X_0) &= \var\left(\frac{1}{2} X_{\varnothing} + \frac{\sqrt{3}}{2} Y_\varnothing\right)\\
             & = \left(\frac{1}{2}\right)^2 \var(X_\varnothing) + \left(\frac{\sqrt 3}{2}\right)^2 \var(Y_\varnothing) =  \frac{1}{4} + \frac 34  = 1.
   \end{align*}
  Similarly $\var(Y_0) = 1$. For the covariance, we have 
  \begin{align*}
       \cov(X_0, X_1) &= \cov\left(\frac{1}{2} X_{\varnothing} + \frac{\sqrt{3}}{2} Y_\varnothing, \frac{1}{2} X_{\varnothing} - \frac{\sqrt{3}}{2} Y_\varnothing\right)\\
             & = \left(\frac{1}{2}\right)^2 \var(X_\varnothing) - \left(\frac{\sqrt 3}{2}\right)^2 \var(Y_\varnothing) =  \frac{1}{4} - \frac 34  = -0.5.
   \end{align*}
This concludes our proof for the $k = 1$ case. Assuming the result holds for \( k = l \), we note that for \( k= l+1 \), the nodes in the left subtree (labelled \( X_{0\star} \)) and the nodes in the right subtree (labelled \( X_{1\star} \)) each constitute a complete binary tree of depth \( l \). Furthermore, the recursive nature of Algorithm \ref{alg:noise-mechanism} indicates that the following three sets of random variables are identically distributed:

\begin{align}\label{eqn:correlation-identity}
   \{X_{0\star}\}_{\star \in \bigcup_{0 \leq s \leq l} \{0,1\}^s} \stackrel{d}{=} \{X_{1\star}\}_{\star \in \bigcup_{0 \leq s \leq l} \{0,1\}^s} \stackrel{d}{=} \{\tilde{X}_{\star}\}_{\star \in \bigcup_{0 \leq s \leq l} \{0,1\}^s}. 
\end{align}

In this context, \( \{X_{0\star}\} \) and \( \{X_{1\star}\} \) represent the left and right subtrees generated by Algorithm \ref{alg:noise-mechanism} when \( k = l+1 \), while \( \{\tilde{X}_{\star}\} \) denotes an independent execution of Algorithm \ref{alg:noise-mechanism} with \( k = l \). Employing equation \eqref{eqn:correlation-identity} and the inductive hypothesis for \( k = l \) on the set \( \{\tilde{X}_{\star}\}_{\star \in \bigcup_{0 \leq s \leq l} \{0,1\}^s} \), we know that the covariance matrices restricted to the left and right subtrees are both equivalent to \( \bsC_l \). Hence, the two \( 2^l \times 2^l \) diagonal blocks of \( \bsC_{l+1} \) are identical to \( \bsC_l \).

To confirm that the entire covariance matrix is equal to \(\bsC_{l+1} \) as specified in \eqref{eqn: corr_recursive}, we must demonstrate that any leaf in the left subtree is correlated with any leaf in the right subtree by \( -2^{-2l-1} \). Take any pair of indices \( I, J \in \{0,1\}^l \), and denote \( \tilde I = (0, I_1, \ldots, I_{l-1}) \in \{0,1\}^l \) and \( \tilde J = (1, J_1, \ldots, J_{l-1}) \in \{0,1\}^l \). Algorithm \ref{alg:noise-mechanism} informs us that the noise at leaf \( X_{0I} \) is expressed as:
\[
X_{0I} = \frac{1}{2} X_{\tilde I} + \frac{(-1)^{I_k}\sqrt{3}}{2} Y_{\tilde I},
\]

and likewise for leaf \( X_{1J} \):
\[
X_{1J} = \frac{1}{2} X_{\tilde J} + \frac{(-1)^{J_k}\sqrt{3}}{2} Y_{\tilde J}.
\]

Consequently:
\begin{align}\label{eqn:corr-leaf-different-side}
    \cov(X_{0I}, X_{1J}) = \left(\frac{1}{2}\right)^2 \cov(X_{\tilde I}, X_{\tilde J}),
\end{align}

since \( Y_{\tilde I} \) is by construction independent of \( X_{\tilde J} \) and  \( Y_{\tilde I} \), and \( Y_{\tilde J} \) are independent of \( X_{\tilde I} \). Given that \( X_{\tilde I} \) and \( X_{\tilde J} \) are leaf nodes on opposing sides of a binary tree of depth \( l \), their covariance is \( -2^{-2k + 1} \) by the inductive hypothesis. Inserting this into equation \eqref{eqn:corr-leaf-different-side}, we obtain:
\begin{align}\label{eqn:final-covariance-leaf-different-sides}
    \cov(X_{0I}, X_{1J}) = 2^{-2} \times -2^{-2k + 1} = -2^{-2k - 1},
\end{align}

which completes the proof.
\end{proof}

\begin{proof}[Proof of Theorem \ref{thm: privacy gaussian}]
To prove $(\eps,\delta)$-DP, we prove that with probability at least $1-\delta$ the likelihood ratio between $\cM_{\sigma}(X)$ and  $\cM_{\sigma}(X')$ is upper bounded by $\exp(\epsilon)$, where $X$ and $X'$ are any two datasets that differ by at most one at only one entry. 
For now, we assume $X = X' + e_i$, where $e_i\in \mathcal{X}^n$ takes value $1$ at the $i$-th coordinate, and $0$ elsewhere. Let $s \in \mathbb R^n$ be any real vector, the logarithm of the likelihood ratio between $\cM_{\sigma}(X)$ and  $\cM_{\sigma}(X')$ at point $X + s$ equals
\begin{align*}
    \log \Bigg(\frac{f_{\cM_{\sigma}(X)} (X + s)}{f_{\cM_{\sigma}(X')} (X + s)}\Bigg) &
    = \log \Bigg( \frac{\exp{(-\frac{1}{2\sigma^2} s^\top \bsC^{-1} s)}}{\exp{(-\frac{1}{2\sigma^2} (s+e_i)^\top \bsC^{-1} (s+e_i))}} \Bigg) \\
    &= -\frac{1}{2\sigma^2} \Big( s^\top \bsC^{-1} s - (s + e_i)^\top \bsC^{-1} (s + e_i) \Big) \\
    &= \frac{1}{2\sigma^2} (\bsC^{-1})(i,i) + \frac{1}{\sigma^2} s^\top \bsC^{-1}e_i,
\end{align*}

In our noise mechanism, $s\sim \bN(\mathbf{0}, \sigma^2 \bsC)$. Therefore the above log-likelihood ratio is a one-dimensional Gaussian with mean $$\frac{1}{2\sigma^2} (\bsC^{-1})(i,i),$$
and variance $$\frac{1}{\sigma^2} e_i^\top \bsC^{-1} \bsC\bsC^{-1} e_i = \frac{1}{\sigma^2} \bsC^{-1}(i,i).$$

Let $Y \sim \bN(0,1)$ be a standard normal random variable, set $w =  \bsC_{K}^{-1}(i,i)/\sigma^2$ for notational simplicity. It is then clear that $\sqrt{w} Y + w/2$ has the same distribution as  $\log \left(\frac{f_{\cM_{\sigma}(X)} (X + s)}{f_{\cM_{\sigma}(X')} (X + s)}\right)$. Since $ \sqrt{w} Y + w/2 \leq \epsilon$ is equivalent to $ Y  \leq  \epsilon/\sqrt{w} - \sqrt{w}/2$, it suffices to find conditions on $w$ such that $\bP[\lvert Y\rvert \geq \epsilon/\sqrt{w}  -\sqrt{w}/2] \leq \delta$.

Applying the standard subgaussian tail bound on $Y$, we have
\[
\bP[Y > x] \leq 2 \exp(-x^2/2)
\]
for every $x > 0$. Therefore it suffices to find $w$ such that $\epsilon/\sqrt{w} - \sqrt{w}/2 > 0$, and 
\[
\left(\frac{\epsilon}{\sqrt{w}} - \frac{\sqrt w}{2}\right)^2 \geq \log\left(\frac{2}{\delta}\right).
\]

Now we set $0 < w \leq \epsilon^2/(2\log(2/\delta))$, and check every $w$ in this range satisfies the above two conditions. Firstly, 
$$w \leq \epsilon^2/(2\log(2/\delta)) \leq \epsilon^2/2\log(4) \leq 2\epsilon,$$
therefore $\epsilon/\sqrt{w} - \sqrt{w}/2 > 0$. Secondly, 
\[
\frac{\epsilon^2}{w} \geq 2\log\left(\frac{2}{\delta}\right) \geq \log(2/\delta) + \log(4) > \log(2/\delta) + \epsilon.
\]
Therefore 
\[
\left(\frac{\epsilon}{\sqrt{w}} - \frac{\sqrt w}{2}\right)^2  = \frac{\epsilon^2}{w} -\epsilon + \left(\frac{\sqrt w}{2}\right)^2 \geq  \frac{\epsilon^2}{w}  - \epsilon \geq \log(2/\delta),
\]
as desired. This immediately translates to the noise bound 
\begin{align}\label{eqn:noise-bound-entry-i}
\sigma^2 \geq \frac{2  (\bsC^{-1})(i,i) \log(2/\delta)}{\epsilon^2}.
\end{align}
The above calculation still goes through when $X = X' - e_i$. Since the preceding argument must be valid for any \( i \), taking the maximum of \( i \) from Equation \eqref{eqn:noise-bound-entry-i} shows our mechanism $\cM_{\sigma}(.)$ is $(\eps,\delta)$-DP for
\[\sigma^2 \geq \frac{2 \lVert \mathsf{diag}(\bsC^{-1}) \rVert_\infty \log(2/\delta)}{\epsilon^2}\] 
\end{proof}

\begin{proof}[Proof of Theorem \ref{thm:privacy}]
   Based on Corollary \ref{cor:inverse diagonal}, it becomes evident that the norm $\lVert \mathsf{diag}(\bsC_k^{-1}) \rVert_\infty$ equals $1 + k/3 = 1+\log_2(n)/3$. Applying this into Theorem \ref{thm: privacy gaussian} leads to  our result.
\end{proof}

\begin{lemma}\label{lem: inverse}
For any positive integer $j$, $\bsC_j^{-1}$ satisfies
$\bsC_1^{-1} :=  \bigl( \begin{smallmatrix}4/3 & 2/3 \\ 2/3 & 4/3\end{smallmatrix}\bigr)$ and
\[
 \bsC_{j}^{-1} =
\begin{pmatrix}[1.5]
    \bsC_{j-1}^{-1} + \frac{1}{3}J_{j-1} & \frac{2}{3} J_{j-1} \\
    \frac{2}{3} J_{j-1} & \bsC_{j-1}^{-1} + \frac{1}{3}J_{j-1},
\end{pmatrix},
\]
 where $J_{j-1}$ is the all one matrix of size $2^{j-1} \times 2^{j-1}$.
\end{lemma}

\begin{proof}[Proof of Lemma \ref{lem: inverse}]
    We prove a slightly stronger result. We show for every $k\geq 1$, the followings all hold:
    \begin{enumerate}
        \item $J_k \bsC_k = \bsC_k J_k = 2^{-k} J_k$
        \item $J_k J_k = 2^k J_k$
        \item $J_k \bsC_{k}^{-1} = 2^k J_k$
        \item We have 
        \begin{align*}
    \bsC_1^{-1} =  \begin{pmatrix}[1.5]
4/3 & 2/3 \\
2/3 & 4/3
\end{pmatrix}
\end{align*}
and 
\begin{align*}
    \bsC_{k}^{-1} =
\begin{pmatrix}[1.5]
   \bsC_{k-1}^{-1} + \frac{1}{3}J_{k-1} & \frac{2}{3} J_{k-1} \\
    \frac{2}{3} J_{k-1} & \bsC_{k-1}^{-1} + \frac{1}{3}J_{k-1} 
\end{pmatrix}.
\end{align*}
     \end{enumerate}
We prove this by induction. All the four claims can be directly checked when $k = 1$. Assuming they are all satisfied when $k \leq K$, for $k = K+1$, facts $1$ and $2$ can still be directly checked. To show fact $4$, it suffices to show 
\begin{align*}
\begin{pmatrix}[1.5]
    A & B \\
    C & D 
\end{pmatrix}  := &
\begin{pmatrix}[1.5]
\bsC_{K} & -2^{-2K - 1}J_K \\
-2^{-2K - 1}J_K & \bsC_{K}
\end{pmatrix} \\ 
& \cdot 
    \begin{pmatrix}[1.5]
    \bsC_{K}^{-1} + \frac{1}{3}J_{K}, & \frac{2}{3} J_{K} \\
    \frac{2}{3} J_{K}, & \bsC_{K}^{-1} + \frac{1}{3}J_{K} 
\end{pmatrix} = I_{2^{K+1}}
\end{align*}

We directly check this by multiplying the two $2\times 2$ block matrices on the left hand side. The diagonal  blocks are
\begin{align*}
A = D & = I + \frac{1}{3} \bsC_K J_K   - \frac{2}{3}\times 2^{-2K-1} J_K J_K \\
& = I + \frac{2^{-K}}{3} J_K - \frac{2^{-K}}{3} J_K = I
\end{align*}
where the second to last  equality follows from facts $1$ and $2$ of our induction hypothesis.

The off-diagonal blocks are:
\begin{align*}
   B = D  &= \frac{2}{3} \bsC_K J_K -2^{-2K-1} J_K \bsC_K^{-1} - \frac{2^{-2K-1}}{3} J_K J_K \\
    & = \frac{4 \times 2^{-K-1}}{3} J_K -2^{-K-1} J_k - \frac{2^{-K-1}}{3} J_k = 0
\end{align*}
This concludes our induction for fact $4$. Fact $3$ is then immediate using fact $4$ together with facts $1,2$ from the induction hypothesis. 
\end{proof}

Corollary~\ref{cor:inverse diagonal} below can be derived directly from Lemma~\ref{lem: inverse} by induction.
\begin{corollary}\label{cor:inverse diagonal}
For any positive integer $k$, the diagonal entries of $\bsC_j^{-1}$ all equal to $1 + j/3$. 
\end{corollary}

\begin{proof}[Proof of Proposition \ref{prop:relationship-errors}]
The first inequality directly follows from the fact that 
\[
\lVert \bsv \rVert_2^2 = \sum_{i=1}^m v_i^2 \leq m \lVert \bsv \rVert_\infty^2 
\]
for any $m$-dimensional vector $\bsv$. 

The second inequality directly follows from the fact that
\[
\lVert \bE[ V] \rVert_\infty \leq \bE[\lVert  V \rVert_\infty ] 
\]
for any random vector $V$. 
\end{proof}

\begin{proof}[Proof of Lemma \ref{lem:maximum-consecutive-variance}, upper bound]
    To prove the upper bound, our main idea is to argue that any fixed summation of the form $\sum_{L = I}^{I+k} X_L$ can be rewritten as a summation $\sum_{l= 1}^T Y_l$, which satisfies the following properties: 
\begin{itemize}
    \item Each $Y_l$ has variance $\sigma^2$
    \item Each pair of $Y_i, Y_j$ has a non-positive correlation
    \item The term of summations $T = O(K)$. 
\end{itemize}
Assuming all the three properties hold, putting these altogether shows
\[
\var[\sum_{L = I}^{I+k} X_L] = \var[\sum_{l= 1}^T Y_l] \leq \sum_{l= 1}^T \var[Y_l] = \sigma^2 T \leq CK \sigma^2. 
\]

The way we construct $\sum_{l= 1}^T \var[Y_l]$ is by defining a ``merge'' operation. Given a continuous summation  $\sum_{L = I}^{I+k} X_L$, we will merge the summations using the equation $X_\star = X_{\star 0} + X_{\star 1}$ for any $\star \in \{0,1\}^k$ for $0\leq k \leq K-1$ as much as possible. For example, summing over all leaves in the left subtree, i.e, $\sum_{L= 000\ldots 0}^{0111\ldots 1} X_L$ can be merged to $X_{0}$. We can merge these summations until no further merging can happen, then we have $\sum_{L = I}^{I+k} X_L = \sum_{l=1}^{T} Y_l$, where each $Y_l$ is the random variable associated with a (not-necessarily leaf) node in the original tree.
It is proven in Theorem \ref{thm:noise-equal-variance} that each $Y_l$ has variance $\sigma^2$, which confirms Property 1 above. The rest two properties are justified by the two lemmas below.
\end{proof}

\begin{lemma}\label{lem: neg-corr}
    Let $S_1, S_2$ be two non-empty subsets of $\{0,1\}^K$ and $S_1 \cap S_2 = \varnothing$. Let $Z_1 := \sum_{i\in S_1} X_i$ and $Z_2 = \sum_{j \in S_2} X_j$, then $\cov(Z_1, Z_2)\leq 0$. Therefore $\cov(Y_l, Y_{l'}) \leq 0$ for any $l\neq l'$. 
\end{lemma}

\begin{proof}[Proof of Lemme \ref{lem: neg-corr}]
    We observe $\cov(Z_1, Z_2) = \sum_{i\in S_1, j\in S_2} \cov(X_i, X_j)$. Since the covariance matrix $\sigma^2\bsC_K$ in \eqref{def:correlated-noise-correlation} has only positive entries on the diagonal, but the above summation does not include the diagonal, we conclude $\cov(Z_1, Z_2)\leq 0$. 
\end{proof}

\begin{lemma}\label{lem:merging-terms}
    Any continuous summation $\sum_{L = I}^{I+k} X_L$ can be merged as a summation $\sum_{l=1}^{T} Y_l$ with $T \leq 2K-2$ when $K\geq 2$.
\end{lemma}
\begin{proof}[Proof of Lemma \ref{lem:merging-terms}]
    We first prove any continuous summation of the form $\sum_{L = 000,\ldots, 0}^{T} X_L$ can be merged into at most $K$ summations. The case where $K = 1$ is straightforward. Suppose $K = m-1$ is proven, when $K = m$, we pick $T$ such that the summation $\sum_{L = 000\ldots 0}^{T} X_L$ has the largest number of summing terms after merging. We may assume without loss of generality that $L \geq 100,\ldots,0$, as otherwise the summation only happens on the left sub-tree which reduces to our inductive hypothesis. Now we can already split the first left-tree in the summation, and write $\sum_{L = 000\ldots 0}^{T} X_L = X_{0} + \sum_{L = 100\ldots 0}^{T} X_L$. The second term has length at most $m-1$ by inductive hypothesis, as the summation happens at the leftmost node of the right subtree. Therefore the total length is at most $m$, as claimed. Moreover, the symmetries of the complete binary tree shows any summation of the form $\sum_{L = T'}^{111,\ldots,1} X_{L}$ can also be merged into at most $K$ summations.

Now we prove Lemma \ref{lem:merging-terms}. The $K = 2$ case is  simple. We assume the claim  holds when $K \leq m-1$. When $K = m$, we pick the consecutive summation which has the largest number of terms after merging.  Again, we may  assume the summation includes both nodes on the left and right subtree, as otherwise it reduces to the $m-1$ case.
Now we write $\sum_{L = I}^{I+k} X_L = W_1 + W_2$, where $W_1$ is the summation on the left part ending at the rightmost node of the left subtree, and $W_2$ on the right part starting at the leftmost node at the right subtree. By our previous induction, both $W_1$ and $W_2$ can be merged into a summation of at most $m-1$ terms. Therefore the whole summation has at most $2m-2$ terms.
\end{proof}

These results together give us the upper bound. Now we turn to the lower bound, we will inductively construct a sequence of continuous range summations for each $K$, and argue the variance scales as $\Theta(K)$.

We first assume $K$ is an odd number, then we consider the consecutive sum by picking the first $M_K := 2^{K-1} + 2^{K-3} + 2^{K-5} + ... + 1$ entries. Intuitively, we are picking all the left subtree, and the left subtree of the right subtree, and the left subtree of the right subtree of the right subtree, and so on. 
\begin{proof} [Proof of Lemma \ref{lem:maximum-consecutive-variance}, lower bound]
    For any odd number $K$, Let $$V_K := \var[ \sum_{k = 000\ldots 0}^{M_K} X_k],$$  
    where $M_K := 2^{K-1} + 2^{K-3} + 2^{K-5} + ... + 1$. Then we claim $V_K - V_{K-2} \geq \frac 23 \sigma^2$. Given the claim,  we have $V_K/\sigma^2 = \Omega(K)$ for odd $K$. Since the maximum variance is a non-decreasing function of $K$, we know 
    \[\max_{I \in \{0,1\}^K, k\leq 2^K - I}\frac{\var[\sum_{L = I}^{I+k} X_L]}{\sigma^2} = \Omega(K).\]

To show the claim, we first merge the summation of the these $M_K$ entries, and write
 \begin{align*}
    \sum_{k = 000\ldots 0}^{M_K} X_k = \sum_{l = 1}^{(K+1)/2} Y_l,
 \end{align*}
 where $Y_1$ is the left subtree, $Y_2$ is the next summation of $2^{K-3}$ entries, and so on. Next we define $W_K := \sum_{l=2}^{(K+1)/2} Y_l$. By the recursive structure of our covariance matrix, $W_K$ has variance $V_{K-2}$. Therefore

 \begin{align*}
      V_n = \var\left[ \sum_{k = 000,\ldots, 0}^{M_K} X_k\right] = \var [Y_1 + W_K] = \sigma^2 + V_{K-2} +2 \cov[Y_1, W_K].
 \end{align*}
The covariance term can  be calculated through \eqref{def:correlated-noise-correlation}, which is $- 2^{-2(K-1) - 1} \sigma^2 |Y_1| |W_K|$. Therefore we can further bound $V_K$ as
\begin{align*}
      V_K &= \var\left[ \sum_{k = 000\ldots 0}^{M_K} X_k\right] = \var [Y_1 + W_K] = \sigma^2 + V_{K-2} +2 \cov[Y_1, W_K] \\
      & = \sigma^2 + V_{K-2} - \sigma^2 2^{-2(K-1)} 2^{K-1} M_{K-2} \\
      &  = \sigma^2 + V_{K-2} - \sigma^2 2^{-(K-1)} (2^{K-3} + 2^{K-5} + \ldots + 1) \\
      & \geq  \sigma^2 + V_{K-2} -\sigma^2 \frac 43 2^{-(K-1)} 2^{K-3} = V_{K-2} +\frac 23 \sigma^2,
\end{align*}
as desired. The last inequality is because 
\begin{align*}
    2^{K-3} + 2^{K-5} + \ldots + 1 &= 2^{K-3}(1 + 1/4 + \ldots  + 2^{-(K-3)})\\
    &\leq 2^{K-3}(1 + 1/4 + \ldots  + 2^{-(K-3)} + \ldots)\\
    & =\frac 43  2^{K-3}.
\end{align*}
 
\end{proof}

\begin{proof}[Proof of Theorem \ref{thm:utility-continuous}]
    Recall the classical fact that $\bE[\lvert X\rvert] = \sigma \sqrt{2/\pi}$ for $X\sim \bN(0,\sigma^2)$, we immediately have:
    \begin{align*}
        \err_{\bsW}^{\infty}(\calW_\sigma) &= \big\Vert\bE_{\bss\sim \bN(\mathbf{0}, \sigma^2 \bsC_K)} [  \lvert \bsW \bss \rvert ]\big\Vert_\infty \\
        & = \sqrt{\frac 2 \pi} \max_{I \in \{0,1\}^K, k\leq 2^K - I}\sqrt{\var[\sum_{L = I}^{I+k} X_L]}\\
        &= \Theta(\sigma \sqrt K)\\
        &= \Theta\left(\sigma \sqrt{\log_2(n)}\right),
    \end{align*}
    which proves the first claim. The second claim follows immediately from the first claim and Proposition \ref{prop:relationship-errors}. 

    For the last claim, recall the Gaussian tail bound $\bP(\lvert X \rvert > t) \leq 2 \exp(-t^2/2\sigma^2)$ given $X \sim \bN(0,\sigma^2)$. We can bound the tail probability of the worst-case error as:
\begin{align*}
    \bP_{\bss\sim \bN(\mathbf{0}, \sigma^2 \bsC_K)} [ \Vert \bsW \bss  \Vert_\infty > t] &= \bP[ \lvert \sum_{L = I}^{I+k} X_L \rvert > t \text{~for some ~} I \in \{0,1\}^K,k] \\
    & \leq \sum_{I,k} \bP[\lvert \sum_{L = I}^{I+k} X_L \rvert > t]\\ 
    & \leq  2\sum_{I,k} \exp\left(\frac{-t^2}{2 \var(\sum_{L = I}^{I+k} X_L)}\right)\\
    & \leq 2\sum_{I,k} \exp\left(\frac{-t^2}{2 \max\var(\sum_{L = I}^{I+k} X_L)}\right) \\
    & \leq 2 \binom{n}{2} \exp\left(\frac{-t^2}{2 c \sigma^2 K}\right)\\
    & \leq n^2 \exp\left(\frac{-t^2}{2 c \sigma^2 K}\right).
\end{align*}
Here the  first inequalty follows from the union bound, and the summation is over all the continuous range queries (and therefore $\binom{n}{2}$ terms). The second inequality uses the Gaussian tail bound. The last few inequalities use Lemma \ref{lem:maximum-consecutive-variance}.

   Now set $t_m = m \sqrt{4\log(n) c\sigma^2 K}$ for every positive integer $m$, we know 
  \begin{align*}
    \bP_{\bss\sim \bN(\mathbf{0}, \sigma^2 \bsC_K)} [ \Vert \bsW \bss  \Vert_\infty > t_m] &\leq n^2 \exp\left(\frac{-t_m^2}{2c\sigma^2K}\right)\\
    & = n^2 \exp\left(\frac{-4m^2\log(n) c\sigma^2 K}{2c\sigma^2K}\right)\\
    & = n^2 \exp\left(-2m^2 \log(n) \right)\\
    & = \frac{1}{n^{2m^2-2}}.
\end{align*}
Therefore we can bound 
\begin{align*}
    \err_{\bsW,\infty}(\calW_\sigma) &= \bE_{\bss\sim \bN(\mathbf{0}, \sigma^2 \bsC_K)} [ \Vert \bsW \bss  \Vert_\infty] \\
    & = \int_0^\infty \bP_{\bss\sim \bN(\mathbf{0}, \sigma^2 \bsC_K)} [ \Vert \bsW \bss  \Vert_\infty > t] \diff t\\
    & \leq t_1 + \sum_{m=1}^\infty \int_{t_m}^{t_{m+1}} \bP_{\bss\sim \bN(\mathbf{0}, \sigma^2 \bsC_K)} [ \Vert \bsW \bss  \Vert_\infty > t]  \diff t \\
    & \leq \sqrt{4\log(n) c\sigma^2 K} + \sqrt{4\log(n) c\sigma^2 K} \sum_{m=1}^\infty \bP_{\bss\sim \bN(\mathbf{0}, \sigma^2 \bsC_K)} [ \Vert \bsW \bss  \Vert_\infty > t_m] \\
    & = \sqrt{4\log(n) c\sigma^2 K}  + \sqrt{4\log(n) c\sigma^2 K}   \sum_{m=1}^\infty \frac{1}{n^{2m^2 - 2}} \\ 
    & = O(\sqrt{4\log(n) c\sigma^2 K}) =  O(\sigma \log(n)) \qquad \text{as~~} K = \log_2(n).
\end{align*}
\end{proof}
\begin{proof}[Proof of Corollary \ref{cor:utility-continuous-privacy}]
    Corollary \ref{cor:utility-continuous-privacy} is immediate when plugging the value of $\sigma$ in Theorem \ref{thm:privacy} into Theorem \ref{thm:utility-continuous}.
\end{proof}

\section{Extra results on nodal queries}
\subsection{Theoretical results}
When $\bsW$ represents  the queries for all the nodes in the binary tree. Our results are summarized below: 
\begin{theorem}\label{thm:utility-nodal}
     Fix $\sigma >0$ and let query matrix $\bsW$ be all the nodal queries, we have:
   \begin{itemize}
      \itemsep-0.2em 
       \item $\err_{\bsW}^{\infty}(\calW_\sigma) = \sqrt{2/\pi}\sigma$,
       \item $\err_{\bsW,2}(\calW_\sigma) = (2n-1)\sigma^2$,
       \item $\err_{\bsW,\infty}(\calW_\sigma) = O(\sigma\sqrt{\log_2(n)})$.
   \end{itemize}
\end{theorem}

\begin{proof}[Proof of Theorem \ref{thm:utility-nodal}]
    The first two properties are follows from direct calculation, and are therefore omitted here. 

    The proof for the worst-case expected error is very similar to the calculation in Theorem \ref{thm:utility-continuous}. Recall that all the leaf nodes are labelled by $X_I$ where $I \in \{0,1\}^K$ which has covariance matrix $\sigma^2 \bsC_K$. The internal nodes are labelled by $X_J$ where $J\in \{0,1\}^k$ for some $k\in \{0,1,\ldots, K-1\}$. Therefore $\err_{\bsW,\infty}(\calW_\sigma) = \bE[\max_{k, J} \lvert X_J \rvert, J \in \{0,1\}^k]$. We can again use the standard Gaussian concentration: 

    \begin{align*}
    \bP[\max_{k, J} \lvert X_J \rvert > t] &= \bP[ \lvert X_J \rvert > t \text{~for some ~} J \in \{0,1\}^k] \\
    & \leq \sum_{J,k} \bP[\lvert X_J \rvert > t]\\ 
    & =  (2n-1) \exp\left(\frac{-t^2}{2 \sigma^2}\right) \qquad \text{total~} 2n-1 \text{~summation terms.}\\
\end{align*}
Now set $t_m = m \sqrt{2\sigma^2 \log(n)}$, we know 
\begin{align*}
\bP[\max_{k, J} \lvert X_J \rvert > t_m] \leq \frac{(2n-1)}{n^{m^2}}.
\end{align*}
Therefore 
\begin{align*}
    \err_{\bsW,\infty}(\calW_\sigma) &= \int_0^\infty \bP[\max_{k, J} \lvert X_J \rvert > t] \diff t \\
    & = \sqrt{2\sigma^2 \log(n)} + \sqrt{2\sigma^2 \log(n)} \sum_{m=1}^\infty \frac{(2n-1)}{n^{m^2}} \\
    & = O(\sigma \sqrt{\log(n)}),
\end{align*}
as announced. 
\end{proof}

\begin{corollary}\label{cor:utility-nodal-privacy}
    Let $X\in \mathcal{X}^n$ be any dataset and let $\cM_{\sigma}(X)  = X + \sigma\bN(\mathbf{0}, \bsC_K)$ be our privacy mechanism, where $\sigma$ is chosen such that the mechanism satisfies $(\eps,\delta)$-DP. Let $\bsW$ be all the continuous range queries, we have:
   \begin{itemize}
      \itemsep-0.2em 
       \item $\err_{\bsW}^{\infty}(\calW_\sigma) = \Theta\left( \sqrt{\log n \log(2/\delta)} \eps^{-1}\right) $,
       \item $\err_{\bsW,2}(\calW_\sigma) = \Theta\left(n^2 \log(n)\log(2/\delta)\eps^{-2}\right)$,
       \item $\err_{\bsW,\infty}(\calW_\sigma) = O\left( \log(n) \sqrt{\log(2/\delta)} \eps^{-1}\right)$.
   \end{itemize}
\end{corollary}


\section{Generalizations}\label{subsec: generalization}

\subsection{Cascade Sampling for Streaming data}\label{subsubsec:streaming}

\begin{algorithm2e}
\caption{Cascade Sampling Algorithm for Streaming Data}\label{alg:noise-mechanism-streaming}
\KwData{Depth of the current binary tree $k \ge 1$, current root noise value $X_0$, variance $\sigma^2$ determined by the privacy budget.}
\KwResult{Noise values $\{X_I\}$ for all nodes $I \in \cup_{0\leq i \leq k-1} \{1\} \times \{0,1\}^i$; new root noise $X'$.}
Sample $Y_\star \sim \bN(0,\sigma^2)$ independently\;
Set $X_{1} := -\frac{1}{2} X_{0} + \frac{\sqrt{3}}{2} Y_\star, \quad X' := X_0 + X_1$\;
\For{each node $\star \in \{1\} \times \{0,1\}^i$ at depth $1\leq i \leq k-1$}{
    Sample $Y_\star \sim \bN(0,\sigma^2)$ independently\;
    Define the noise values for the children of $\star$: \;
    $X_{\star 0} := \frac{1}{2} X_{\star} + \frac{\sqrt{3}}{2} Y_\star, \quad X_{\star 1} := \frac{1}{2} X_{\star} - \frac{\sqrt{3}}{2} Y_\star$\;
    }
\end{algorithm2e}


\subsection{General binary trees}\label{subsubsec:general_binary_tree}
Our noise mechanism is designed under the assumption that our data points can be viewed as leaves of a perfect binary tree. Nevertheless, thank to the top-down nature of our cascade sampling algorithm (Algorithm \ref{alg:noise-mechanism}), it can be readily adapted to work with a general binary tree. We will still use binary bits to represent data points. They are still leaves of a binary tree of maximum depth $k$, but not necessarily have the same depth (i.e., same length in the bit string). In our binary tree, the root node is still indexed by $\varnothing$. For any non-leaf node identified by the bit string $\star$, its children are indexed as $\star0$ and $\star1$ if it has two children, or solely as $\star0$ if it has only one child.

Our adapted algorithm closely mirrors Algorithm \ref{alg:noise-mechanism}, with the primary distinction being the evaluation of the number of children at each step. Starting from the root, the process proceeds downward. Given the assigned value to a node $\star$, the algorithm first determines the number of $\star$'s children. If $\star$ has two children, it generates correlated noises in the same way as  Algorithm \ref{alg:noise-mechanism} (or as described in formula \ref{eqn: one-d noise}). If $\star$ has only one child, the same noise  is assigned to $\star0$. If $\star$ is itself a leaf, the algorithm simply moves on to the next node. Detailed description is in Algorithm \ref{alg:noise-mechanism-general-binary-tree} below.

\begin{algorithm2e}
\caption{Noise Allocation Mechanism for a General Binary Tree}\label{alg:noise-mechanism-general-binary-tree}
\KwData{Binary tree with maximum depth $k$, variance $\sigma^2$ determined by the privacy budget.}
\KwResult{Noise values $\{X_I\}$ for all nodes $I \in \cup_{0\leq l \leq k} \{0,1\}^l$ in a binary tree.}
\For{each node $\star \in \{0,1\}^i$ at depth $0\leq i \leq k-1$}{
    \If{$\star =\varnothing$}{
    Assign $X_\varnothing \sim \bN(0,\sigma^2)$\\
    Sample $Y_\star \sim \bN(0,\sigma^2)$ independently\\
    Define the noise values for the children nodes of $\star$:\\
     \If{$\star$ \text{has two children}}{
        $X_{\star 0} := \frac{1}{2} X_{\star} + \frac{\sqrt{3}}{2} Y_{\star}$
     }
    $X_{\star 1} := \frac{1}{2} X_{\star} - \frac{\sqrt{3}}{2} Y_\star$\\
    \ElseIf{$\star$ \text{has one child}}{
    $X_{\star0} = X_\star$
        }
    }
}
\end{algorithm2e}

Algorithm \ref{alg:noise-mechanism-general-binary-tree} continues to ensure that the marginal distributions for each node adhere to $\bN(0,\sigma^2)$. Moreover, the noise level for any given node is equivalent to the sum of the noise levels of its children, thereby maintaining consistency. The theoretical analysis for this scenario closely parallels that of the perfect binary tree case discussed in Section \ref{sec:corr-noise}. Extending this approach to non-binary trees is also possible, though more complicated. This further exploration will be reserved for our future studies.

\subsection{Two-dimensional range queries and contingency tables}\label{subsubsec:2d-queries}
Remember that our approach, as outlined in Algorithm \ref{alg:noise-mechanism}, is appropriate for data arranged in one-dimensional arrays that have hierarchical structures. We are now expanding this method to apply to two-dimensional contingency tables. An example of this is that each data point represents the population of a village, positioned in two dimensions for longitude and latitude. The higher-level hierarchies might represent the populations of larger entities like cities, provinces, or countries.

Putting it formally, we consider data points labeled as $\{X_{I,J}\}_{I\in \{0,1\}^{k_1}, J\in \{0,1\}^{k_2}}$. It is convenient to think of these data points as elements of a matrix with $2^{k_1}$ rows and $2^{k_2}$ columns. Each row and column corresponds to the leaf nodes of a perfect binary tree of heights $k_1$ and $k_2$, respectively. Each pair of (internal) nodes $(I_1, J_1)$ from the two binary trees represents a submatrix. For instance, the pair $(\varnothing, \varnothing)$ corresponds to the entire data matrix. We aim to develop a noise mechanism where the summation of noises applied to each submatrix, as described by $(I_1, J_1)_{I_1 \in \{0,1\}^{i_1}, I_2\in \{0,1\}^{i_2}, i_1\leq k_1, i_2\leq k_2}$, maintain the same variance.

Remember that for the one-dimensional scenario, the essence of our design hinges on the observation that if $X, Y \sim \bN(0,1)$, then the variables $X_0, X_1$ given by 
\begin{align*}
    X_0 = \frac{1}2 X + \frac{\sqrt 3}{2} Y,\quad X_1 = \frac 12 X - \frac{\sqrt 3}{2} Y.
\end{align*}
also follow a $\bN(0,1)$ distribution. The above construction can be generalized to higher dimensions. For example, let $X, Y, Z, W$ be i.i.d. $\bN(0,1)$ random variables, and define:
\begin{align*}
    &X_0 = \frac{1}2 X + \frac{\sqrt 3}{2} Y,\quad X_1 = \frac 12 X - \frac{\sqrt 3}{2} Y \\
    &\tilde X_0 = \frac{1}2 Z + \frac{\sqrt 3}{2} W, \quad \tilde X_1 = \frac{1}2 Z - \frac{\sqrt 3}{2} W.
\end{align*}

Then we can construct 
\begin{align*}
   & (X_{00}, X_{01}) = \frac{1}{2} \left(X_0, X_1\right) + \frac{\sqrt 3}{2} \left(\tilde X_0, \tilde X_1\right) \\
   & (X_{10}, X_{11}) = \frac{1}{2} \left(X_0, X_1\right) - \frac{\sqrt 3}{2} \left(\tilde X_0, \tilde X_1\right),
\end{align*}
or equivalently
\begin{align*}
    &X_{00} = \frac{1}4 X + \frac{\sqrt 3}{4} Y + \frac{\sqrt 3}{4} Z + \frac{ 3}{4} W \\
    &X_{01} = \frac{1}4 X - \frac{\sqrt 3}{4} Y + \frac{\sqrt 3}{4} Z - \frac{ 3}{4} W \\
    & X_{10} = \frac{1}4 X + \frac{\sqrt 3}{4} Y - \frac{\sqrt 3}{4} Z - \frac{ 3}{4} W \\
    &X_{11} = \frac{1}4 X - \frac{\sqrt 3}{4} Y - \frac{\sqrt 3}{4} Z + \frac{ 3}{4} W.
\end{align*}
We can verify that the matrix
\begin{align*}
    \begin{pmatrix}[1.5]
X_{00}& X_{01} \\
X_{10} & X_{11}
\end{pmatrix}
\end{align*}
satisfies the following: 
\begin{itemize}
    \item Each entry is a standard normal
    \item The summations of each row, each column are all standard normals
    \item The summation of all entries are standard normals.
    \item $X_{00} + X_{10} = X_0, X_{01} + X_{11} = X_1$.
\end{itemize}

Building on this concept, we will now describe how to extend our method (Algorithm \ref{alg:noise-mechanism}) for creating private noise mechanisms suitable for $2$-dimensional range queries. Initially, we'll utilize Algorithm 1 to sample the summation of all $2^{k_2}$ columns, meaning the noise values indexed as $(\varnothing,J)$. Following that, we will implement the aforementioned strategy to divide each of these column totals into smaller groups of $2, 4, 8, \ldots, 2^{k_1}$ values. The detailed algorithm is described below.  

\begin{algorithm2e}
\caption{Noise Allocation Mechanism for two-dimensional range queries}\label{alg:noise-mechanism-2d}
\KwData{Depths $K_1,K_2$ of the binary tree indexing rows and columns, variance $\sigma^2$ determined by the privacy budget.}
\KwResult{Noise values $\{X_{I,J}\}$ for all node pairs $(I,J) \in \cup_{0\leq k_1 \leq K_1} \{0,1\}^{k_1}\times \cup_{0\leq k_2 \leq K_2} \{0,1\}^{k_2}$}

Fix $I = \varnothing$, implement Algorithm \ref{alg:noise-mechanism} to assign $X_{\varnothing J}$ for every $J \in  \cup_{0\leq k_2 \leq K_2} \{0,1\}^{k_2}$\;

\For{each node $\star \in \{0,1\}^{k_1}$ for some $0\leq k_1 \leq K-1$}{
    Implement Algorithm \ref{alg:noise-mechanism} to assign $Y_{\star J}$ for every $J \in  \cup_{0\leq k_2 \leq K_2} \{0,1\}^{k_2}$\;
    Define the noise values for the children nodes of $\star$:\;
    $X_{\star 0, J} := \frac{1}{2} X_{\star, J} + \frac{\sqrt{3}}{2} Y_{\star,J}$ for every $J$\;
    $X_{\star 1, J} := \frac{1}{2} X_{\star,J} - \frac{\sqrt{3}}{2} Y_{\star,J}$ for every $J$\;
}
\end{algorithm2e}

It is not hard to verify from the construction that every $X_{I,J}$ has the same marginal distribution, and satisfies $X_{I,J} = X_{I0,J} + X_{I1,J}$ and  $X_{I,J} = X_{I,J0} + X_{I,J1}$. The computational cost of Algorithm \ref{alg:noise-mechanism-2d} is also linear with the number of data points. More detailed analysis of the appropriate level for the privacy budget, along with an analysis of utility, will be deferred to our future research.

\section{Additional experiments}
\label{app:add_exp}

In this section, we demonstrate numerical experiments for when the workload matrix consists of node and random queries and also runtime comparisons.

\subsection{Efficiency of Cascade Sampling }

To justify the practical efficiency of the Cascade Sampling  algorithm described in Section~\ref{subsec:algorithm-1d}, we compare it against \texttt{Scipy}'s built-in function for generating multivariate Gaussian distributions. On a personal laptop, we use both methods to produce Gaussians with a covariance matrix \( C_k \) for \( k \in \{3,4, \ldots, 25\} \), which corresponds to dimensionalities ranging from \( 2^3 = 8\) to \( 2^{25} \approx \) $33.5$ million. In comparison to our sampling algorithm, \texttt{Scipy}'s  \texttt{multivariate\_normal} function consistently takes  longer time. Its efficiency diminishes at higher dimensions. For example, it takes over $28$ minutes to sample a $16,384$ (\(2^{14}\))-dimensional  Gaussian, and more than $3$ hours for a $32,768$  (\(2^{15}\))-dimensional Gaussian, making it impractical for larger dimensions. Samplers from other libraries like \texttt{Numpy} were also tested but could not execute for dimensions exceeding $1024$.

\subsection{Node queries}

Consider a binary tree such that elements of $\bsx$ form the leaves of the tree. The queries are the nodes (including leaves) of this tree, which correspond to an interval summing up the leaves of tree rooted at the particular node. 
The error values are shown in \cref{fig:node_queries}.

\begin{figure}[!htbp]
    \centering
      \includegraphics[scale=0.215]{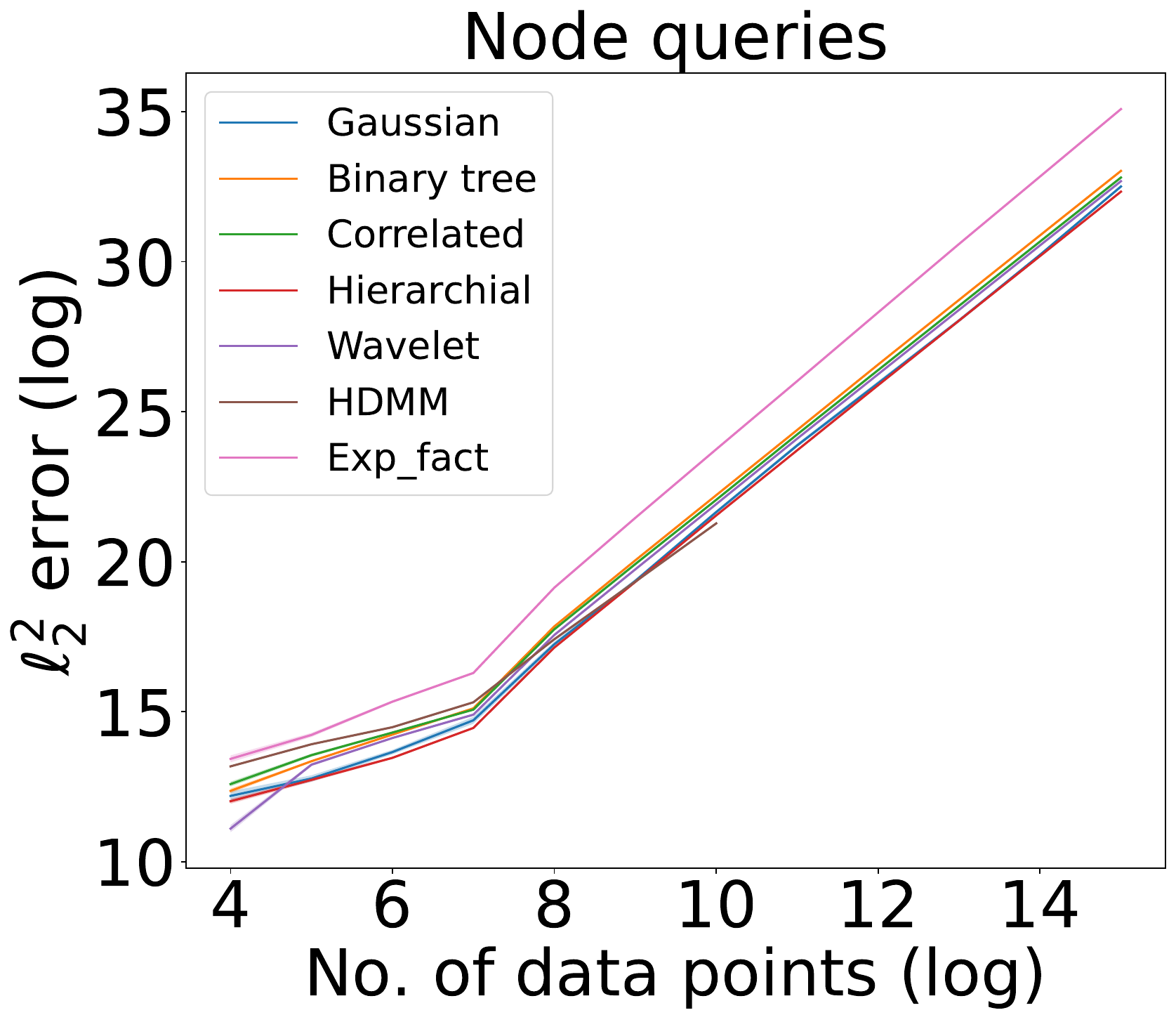}
      \includegraphics[scale=0.2]{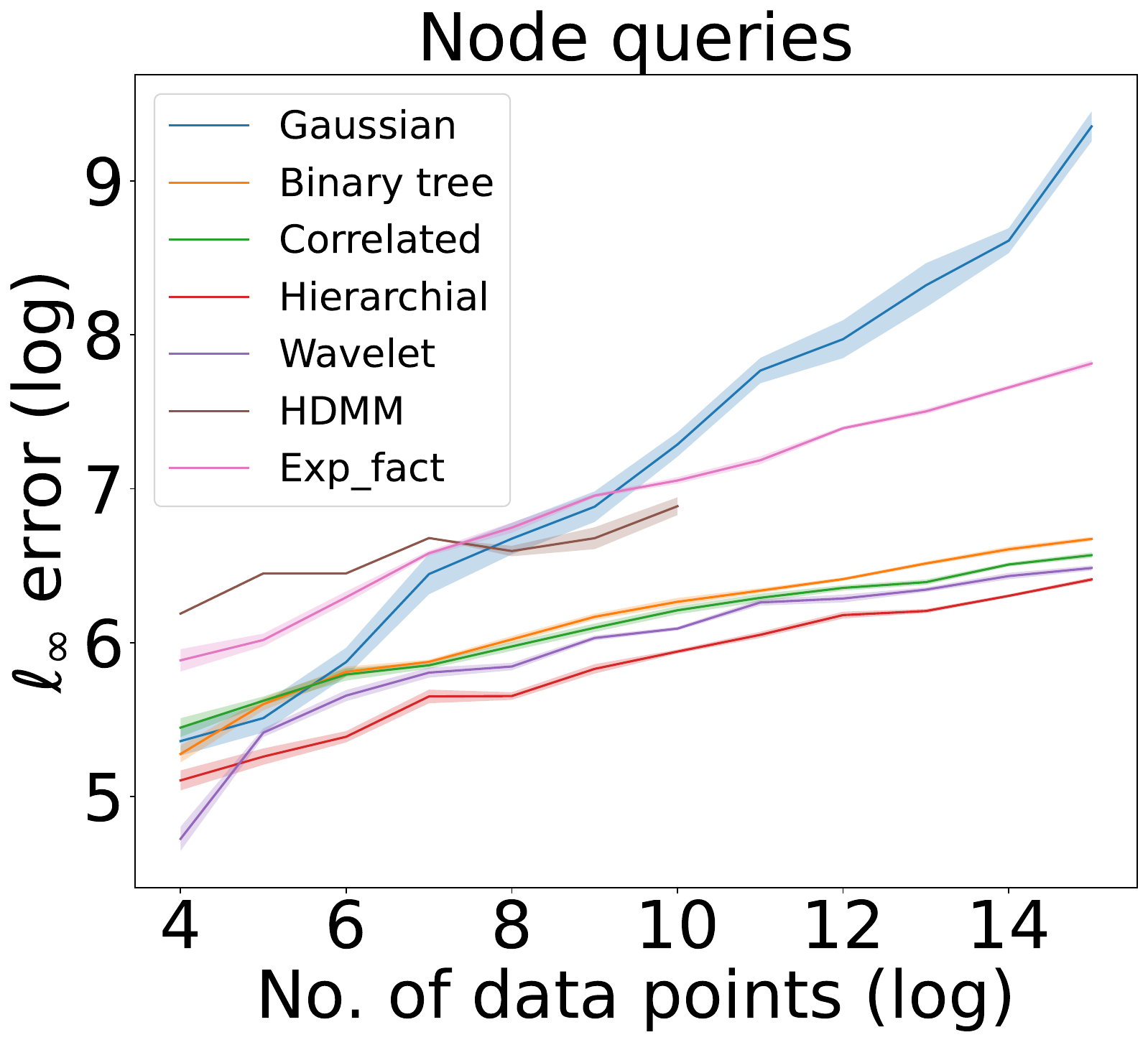}
    \caption{$\err_{\bsW,2}(\calW_\sigma)$ and $\err_{\bsW,\infty}(\calW_\sigma)$ for node queries.}
    \label{fig:node_queries}
\end{figure}

The gap in utility across mechanisms for $\err_{\bsW,2}(\calW_\sigma)$ is small with Hierarchial, Wavelet and ours providing smallest error values, then closely followed by remaining mechanisms. But that is not the case for $\err_{\bsW,\infty}(\calW_\sigma)$, where Gaussian, HDMM and explicit factorization perform significantly worse. The remaining mechanisms perform very similarly (except for Binary Tree which has slightly larger error) as $n$ increases. The gap in correlated input perturbation and Binary tree mechanism is easier to see which comes from the smaller constant for noise that is required for guaranteeing privacy ($ 2+ \frac{2}{3}\log n \hspace{0.25em} \text{vs} \hspace{0.25em} \log n$).

\subsection{Random queries}
\label{app:exp_random}

Here, each query asks for sum of random elements from $\bsx$. Note that unlike first two settings, the query is not contiguous over elements in $\bsx$. We sample the queries as follows - for each query we sample a number $k \in \{n/4, \cdots, n\}$ (to simulate dense queries). We then sample $k$ indices uniformly at random from $[n]$ and these indices form the query (these indices are set to 1 for the particular row in $\bsW$). We fix number of queries to be 2500. The results are shown in \cref{fig:random_queries}

\begin{figure}[!htbp]

    \centering
    \subfigure{
      \includegraphics[scale=0.215]{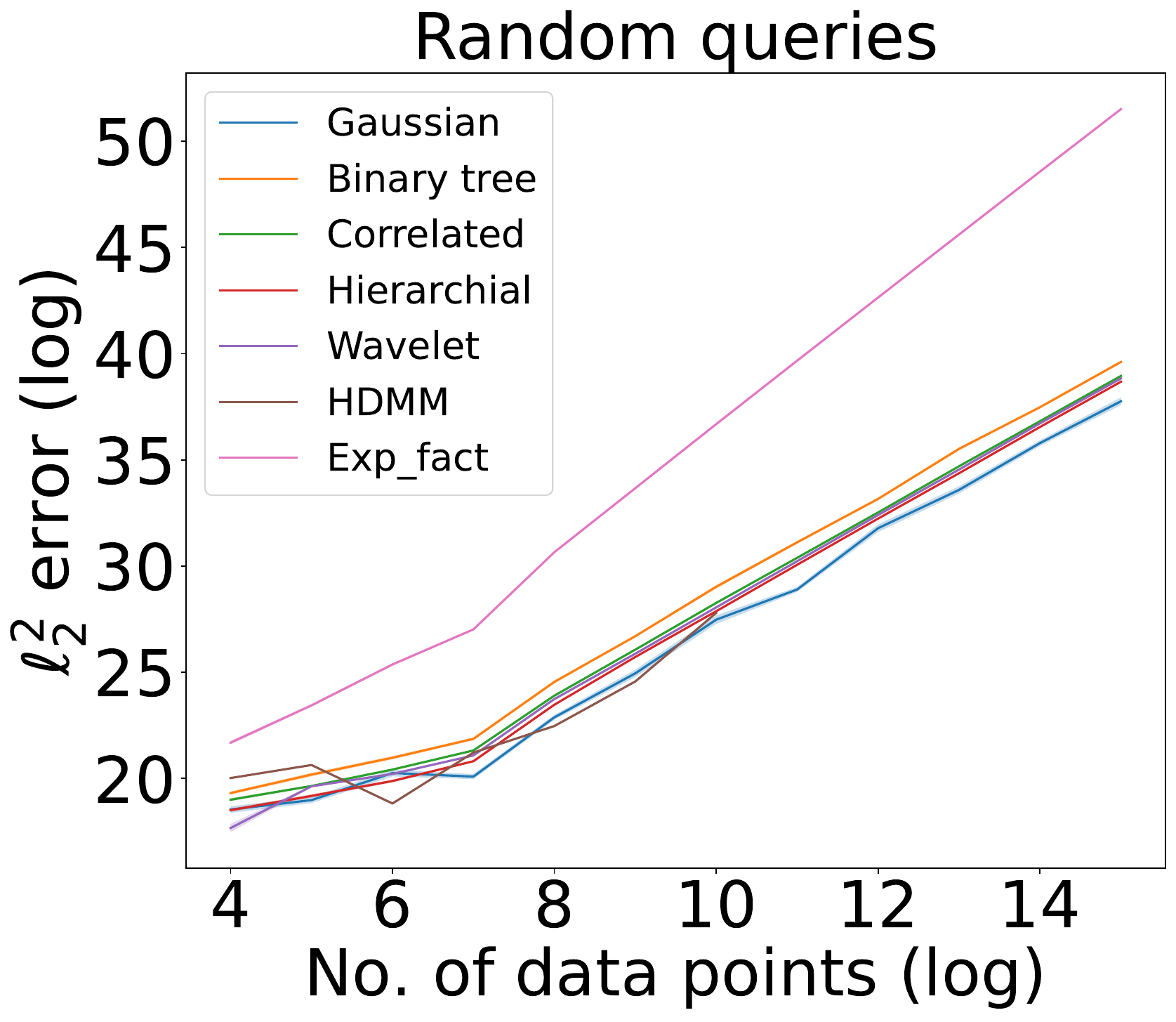}
      \label{fig:random_l2}
    }
    \subfigure{
      \includegraphics[scale=0.2]{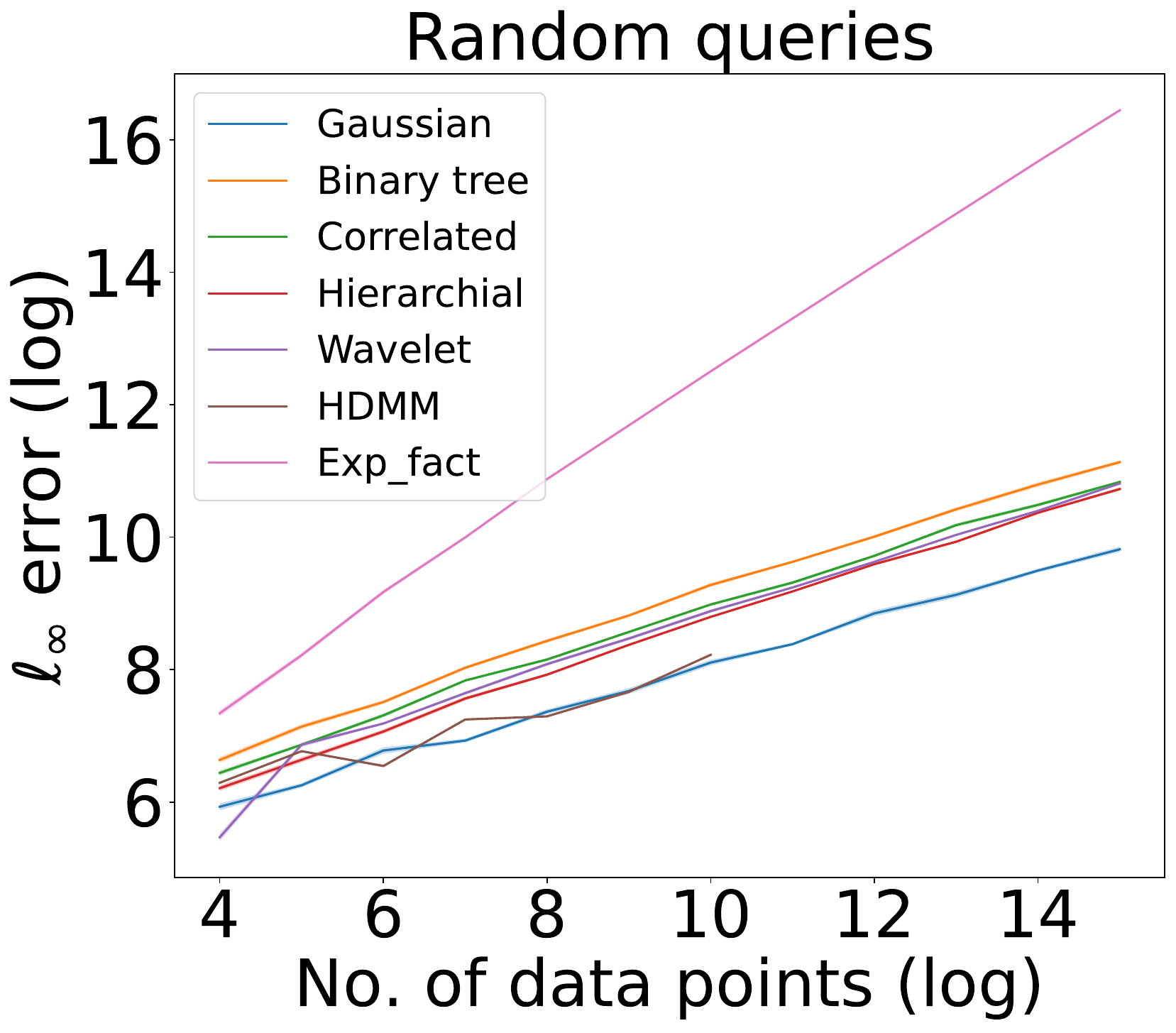}
      \label{fig:random_linf}
    }
    \caption{$\err_{\bsW,2}(\calW_\sigma)$ and $\err_{\bsW,\infty}(\calW_\sigma)$ for random queries.}
    \label{fig:random_queries}
\end{figure}

For this setting of workload matrix, HDMM and Gaussian mechanism perform best utility wise for both $\err_{\bsW,2}(\calW_\sigma)$ and $\err_{\bsW,\infty}(\calW_\sigma)$. Wavelet, ours and Hierarchical perform almost similarly w.r.t noise magnitude whereas explicit factorization mechanism has the largest noise. Explicit factorization uses a very special decomposition tailored to continual counting hence it has worse performance and Hierarchical, ours and Binary tree mechanisms are more suited for sum over contiguous arrays by design, and as a result for random queries they perform slightly worse. 
Gaussian mechanism has the least error majorly due to the fact that the noise scale is smallest due to smaller sensitivity. To see this, note that the queries sum over $O(n)$ elements in $\bsx$. For Gaussian, this corresponds to sum of $O(n)$ i.i.d Gaussians, whereas for Hierarchical and Correlated mechanism, worst case we still sum over $O(n)$ elements but with higher variance (by a factor of $O(\log{n})$), resulting in larger noise added to the query answers.

\subsection{Runtime comparison}
\label{app:exp_runtime}
We compare runtime for different mechanisms in seconds for answering the queries provided by the given workload matrix.
For runtime, Gaussian and Wavelet are typically the fastest across different settings of workload matrix and $n$ whereas our mechanism is competitive with these and faster than Hierarchical. HDMM typically takes the longnest which is expected as it involves expensive optimization to achieve the matrix decomposition. The time required for explicit factorization mechanism majorly comes from answering queries, as their decomposition is fast but for each query, the answer has to be constructed via difference of two other queries.

\begin{figure}[b]
\label{fig:runtime}
    \centering
    \subfigure{
      \includegraphics[scale=0.125]{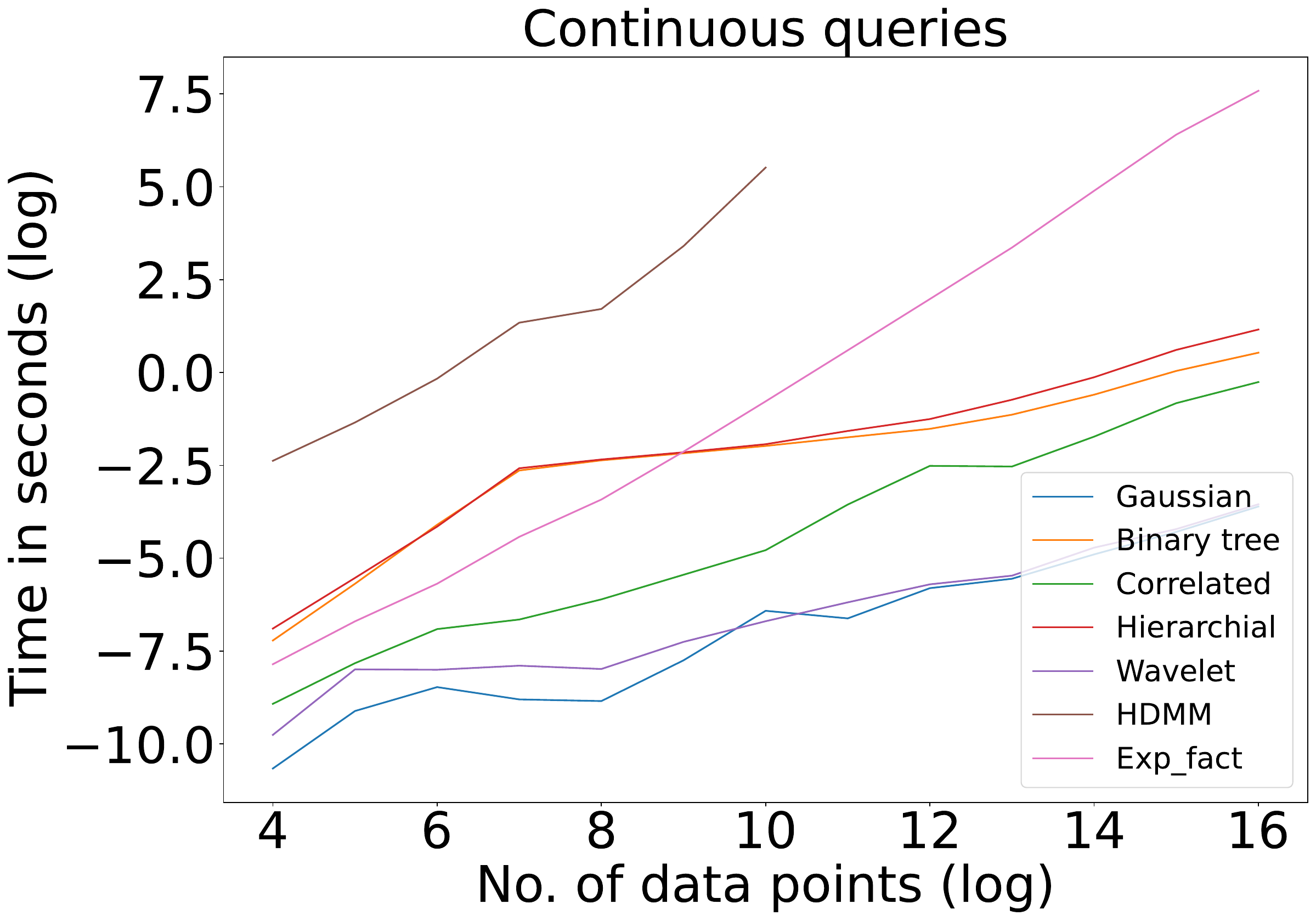}
      \label{fig:time_cont}
    }
    \subfigure{
      \includegraphics[scale=0.125]{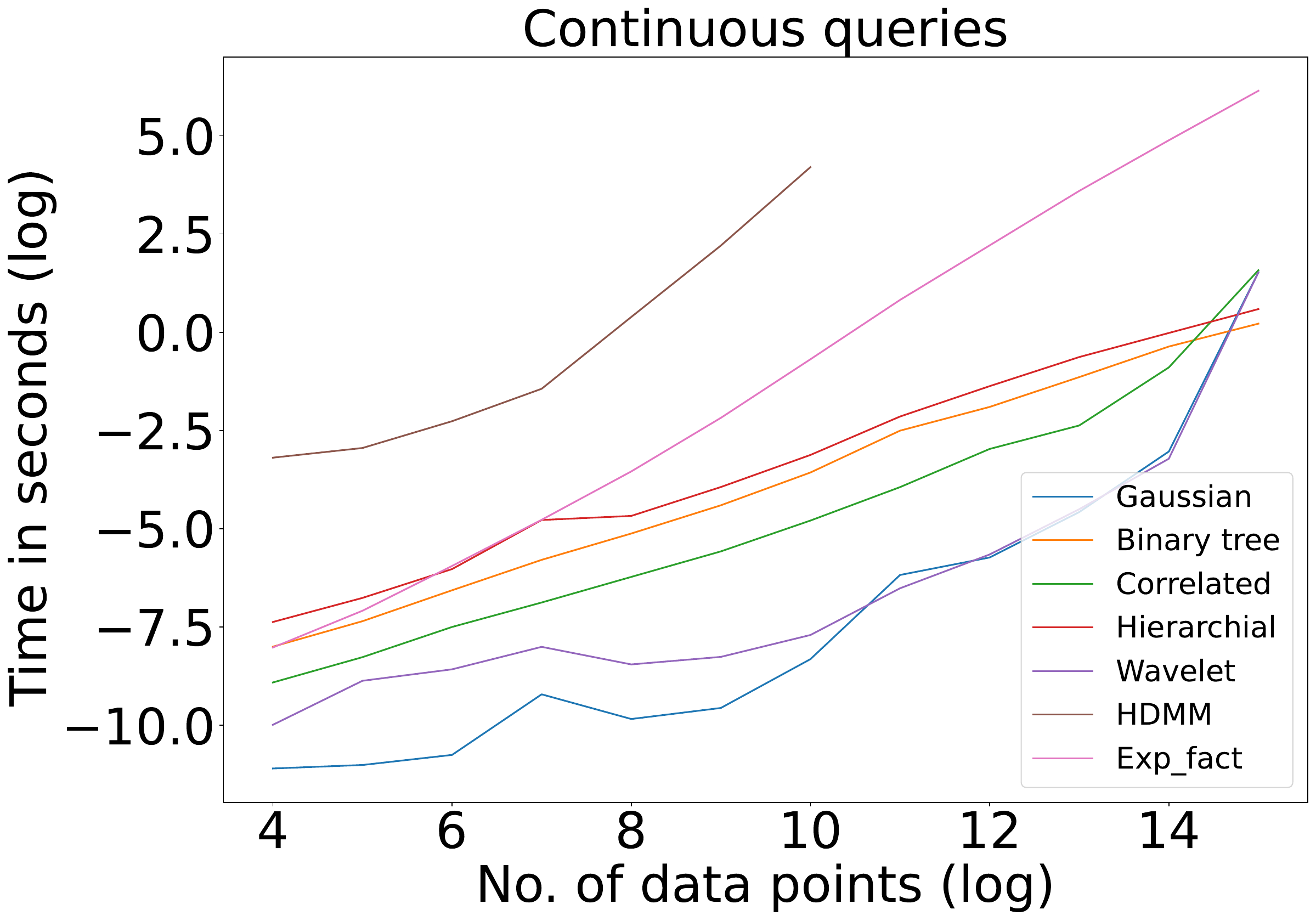}
      \label{fig:time_node}
    }
    \subfigure{
      \includegraphics[scale=0.125]{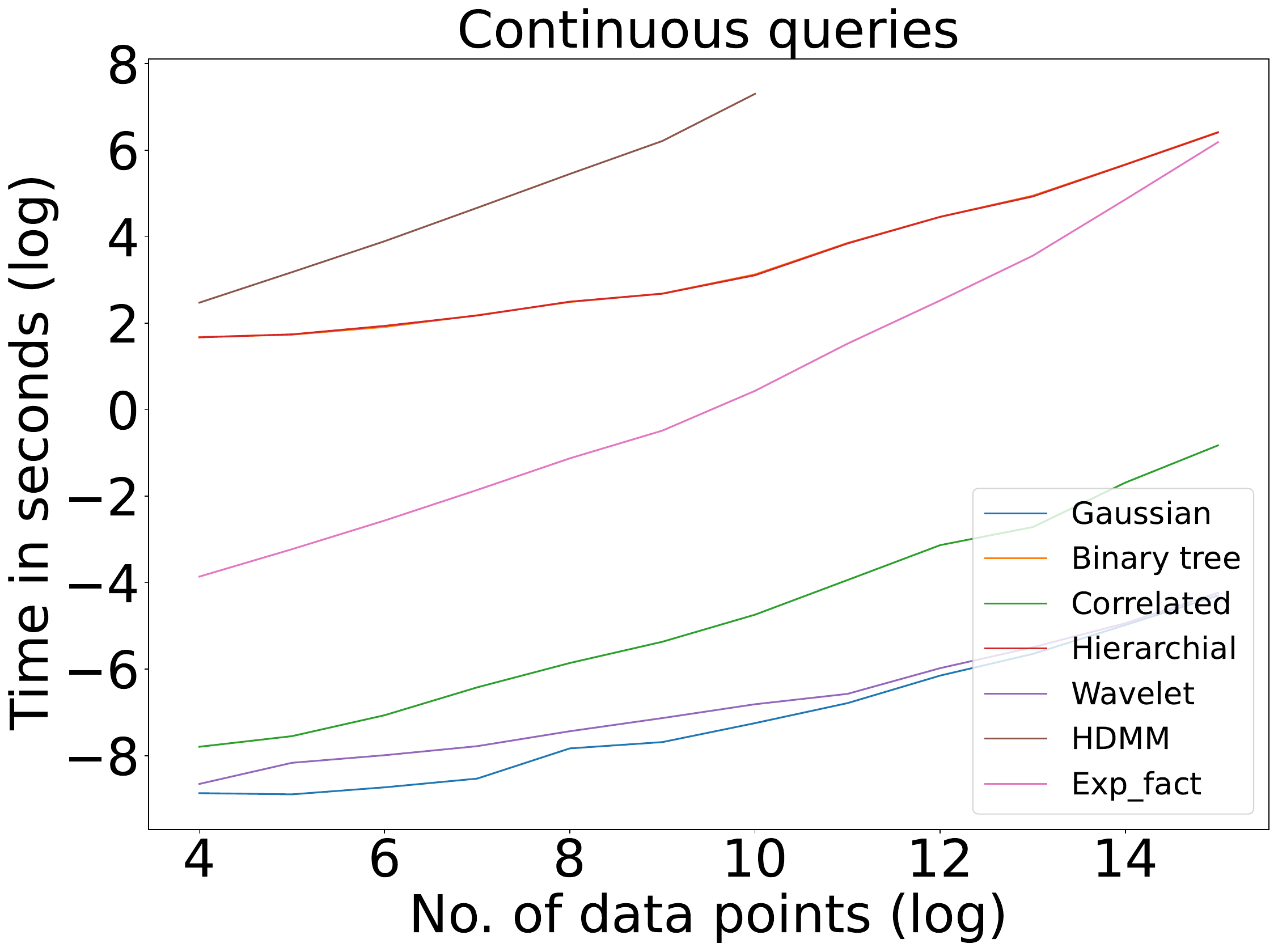}
      \label{fig:time_random}
    }
    \caption{Runtime comparison for different configurations of the workload matrix.}
\end{figure}


\subsection{Implementation Details}\label{subsec:exp_details}

We first describe the implementation for all mechanisms used in evaluations. All experiments were performed on a Macbook Pro with M1 processor and 16GB of RAM. Code is implemented in python language.

\paragraph{Binary Tree:}
Here the mechanism builds a binary tree with elements being the leaves of the tree, and appropriately scaled Gaussian noise (for preserving privacy) is added to the the counts of nodes. We use the segment tree package for building and answering queries.

\paragraph{Hierarchical and Wavelet:}
Hierarchical mechanisn follows the same methodology as the Binary tree mechanism, followed by post-processing to find consistent sums for the internal nodes. We use the implementation from dpcomp-core package \cite{HayMMCZB16} for Hierarchical and Wavelet.

\paragraph{HDMM:}

We use the implementation from dpcompcore library, using the default parameter for number of restarts.

\paragraph{Explicit factorization}

\cite{FichtenbergerHU23} provides implicit matrix factorization for continual counting, i.e., counting of ranges $[1, i]$, for all i. Notice that this is a special case of range queries $[i,j]$ for all . One can easily retrieve the query answer $[i, j]$ by calculating the privatized query answer for $[1, j]$ minus that of $[1, i-1]$. Privacy is preserved (post-processing) although variance for query of $[i, j]$ can possibly increase (at most twice).

\paragraph{Gaussian and Correlated mechanism:}
Implementing these mechanisms is straight forward by adding appropriately scaled noise for privacy and then using the private histograms to answer queries.

\paragraph{Monte Carlo sampling for estimating errors of continuous queries:}

Now we discuss some implementation details of the continuous range queries in Section \ref{sec:experiments}. After implementing the Cascade Sampling algorithm for the correlated perturbation $\cM_{\sigma}(\bsx) = \bsx + \bN(\mathbf{0}, \sigma^2 \bsC_k)$, calculating the error introduced by \textit{all} the $\Theta(n^2)$ continuous range queries  requires a cost of $O(n^3)$ to sum up all the noises in the corresponding entries, which is prohibitive when $n$ is large.  

Our primary goal is to evaluate the errors $\err_{\bsW,2}(\calW_\sigma)$ and $\err_{\bsW,\infty}(\calW_\sigma)$ in comparison to current algorithms, so we employ the Monte Carlo method for estimating these errors. It is clear from the mean and standard deviation plot presented in Figure \ref{fig:results} that our estimation is sufficiently accurate. 

We now outline the implementation process for $n \geq 2^8$. In each iteration, we uniformly select one range from all $\binom{n}{2} + n$ queries. This step is repeated $m$ times. Using the $m$ sampled errors $E_1, \ldots, E_m$, we estimate $\err_{\bsW,2}(\calW_\sigma)$ and $\err_{\bsW,\infty}(\calW_\sigma)$. For instance, $(\binom{n}{2} +n)\times \sum_{i=1}^m (E_i^2)/m$ serves as an unbiased estimator for $\err_{\bsW,2}(\calW_\sigma)$, becoming more accurate as $m$ increases. Additionally, $\max_i E_i$ consistently estimates $\err_{\bsW,\infty}(\calW_\sigma)$
.

In further details, to uniformly select a continuous range, our sampling method involves: 1. Flip a coin with head probability $2/(n+1)$. 2. If heads, then random choose one element in $\{1,\ldots, n\}$. 3. If tails, uniformly choose two distinct elements and use $\min\{n_1,n_2\}$ to $\max\{n_1, n_2\}$ as our range. 

A final remark is that  one could compute all $\binom{n}{2} + n$ errors without Monte Carlo methods using faster than $O(n^3)$ algorithms  by properly storing intermediate results. However, these are not implemented in our current work.



\end{document}